\documentclass[a4paper,UKenglish,cleveref,autoref,thm-restate]{lipics-v2021}

\pdfoutput=1 
\hideLIPIcs  

\title{On Estimating Degree Distribution in Streaming and Query Models} 

\author{Arijit Bishnu}{Indian Statistical Institute, Kolkata, India \and \url{https://www.isical.ac.in/~arijit/}}{arijit@isical.ac.in}{https://orcid.org/0000-0003-0018-7314}{}

\author{Debarshi Chanda}{Indian Statistical Institute, Kolkata, India \and \url{https://sites.google.com/view/debarshi-ch/}}{debarshi.chanda.1997@gmail.com}{https://orcid.org/0009-0001-2197-6366}{}

\author{Gopinath Mishra}{Institute of Mathematical Sciences, Chennai, India \and \url{https://sites.google.com/view/gopinathmishra}}{gopinath@imsc.res.in}{https://orcid.org/0000-0003-0540-0292}{}

\authorrunning{A. Bishnu, D. Chanda, G. Mishra} 



\category{} 

\relatedversion{} 

\nolinenumbers 

\EventEditors{John Q. Open and Joan R. Access}
\EventNoEds{2}
\EventLongTitle{42nd Conference on Very Important Topics (CVIT 2016)}
\EventShortTitle{CVIT 2016}
\EventAcronym{CVIT}
\EventYear{2016}
\EventDate{December 24--27, 2016}
\EventLocation{Little Whinging, United Kingdom}
\EventLogo{}
\SeriesVolume{42}
\ArticleNo{23}

\usepackage{float}

\usepackage{mathtools, amsthm, nicefrac}
\usepackage{amssymb}
\usepackage{thmtools} 
\usepackage{algorithm, times}
\usepackage{2026_ICML/algorithmic}
\usepackage{tabularx}
\usepackage{multirow}
\usepackage{multicol}
\usepackage{wrapfig}
\usepackage{soul}
\usepackage{caption}
\usepackage{subcaption}

\usepackage{thm-restate}

\usepackage{multicol} 

\usepackage{todonotes}

\usepackage{hyperref}
\hypersetup{
    colorlinks=true,
    linkcolor=blue,
    filecolor=blue,      
    urlcolor=blue,
    citecolor=blue,
    pdftitle={Notes - Topic},
    pdfpagemode=FullScreen,
    }

\usepackage[most]{tcolorbox}
\newtcolorbox{idea}[1][]
{
colbacktitle=cyan,
colback=cyan!10,
arc=1pt,
boxrule=1pt,
title=#1 
}

\newtcolorbox{update}[1][]
{
colbacktitle=gray,
colback=gray!10,
arc=1pt,
boxrule=1pt,
title=#1 
}

\newtcolorbox{question}[1][]
{
coltitle=black,
colbacktitle=yellow,
colback=yellow!10,
arc=1pt,
boxrule=1pt,
title=#1 
}

\newtcolorbox{notebox}[1][]
{
coltitle=black,
colbacktitle=green,
colback=green!10,
arc=1pt,
boxrule=1pt,
title=#1 
}

\newtcolorbox{problem}[1][]
{
coltitle=black,
colbacktitle=red!60,
colback=red!10,
arc=1pt,
boxrule=1pt,
title=#1 
}

\newcommand{\bI}{\boldsymbol{I}}

\newcommand{\fN}{\field{N}}

\newcommand{\sA}{\mathcal{A}}

\newcommand{\sD}{\mathcal{D}}
\newcommand{\sE}{\mathcal{E}}

\newcommand{\sR}{\mathcal{R}}

\usepackage{cleveref}

\newcommand{\ccdh}{C}

\newcommand{\approxccdh}{\widehat{\ccdh}}

\newcommand{\approxdom}{\approxerror_\domain}
\newcommand{\approxran}{\approxerror_\ran}
\newcommand{\biapprox}{\fbrac{\approxdom,\approxran}}

\newcommand{\ccdhthreshold}{\tau_\ccdh}
\newcommand{\vertexsamplesize}{q}
\newcommand{\vertexsample}{\vertexset_\vertexsamplesize}
\newcommand{\degreesample}{D^{\vertexsamplesize}}
\newcommand{\degreethreshold}{\tau_{\degree{}}}
\newcommand{\edgesamplesize}{r}
\newcommand{\edgesample}{\edgeset_\edgesamplesize}
\newcommand{\approxdegree}[1]{\widehat{d}_{#1}}

\newcommand{\approxhindex}{\hindex'}
\newcommand{\esthindex}{\hat{\hindex}}

\newcommand{\stream}{\sigma}
\newcommand{\streamelement}{\edge}
\newcommand{\vertexsampleindex}{\bI_{\vertex}}

\newcommand{\streamspace}{S}

\newcommand{\hashrange}{\vertexsamplesize_{a}}
\newcommand{\hash}{h}

\newcommand{\actvercount}{\vertexcount_a}
\newcommand{\actverset}{\vertexset_a}

\newcommand{\medcount}{t}

\newcommand{\vertexsetA}{A}
\newcommand{\vertexsetC}{C}
\newcommand{\vertexsetB}{B}
\newcommand{\vertexsetD}{D}
\newcommand{\vertexsetE}{E}
\newcommand{\isovertexset}{I}
\newcommand{\vertexA}{a}
\newcommand{\vertexB}{b}
\newcommand{\vertexC}{c}
\newcommand{\vertexD}{d}
\newcommand{\vertexE}{e}

\usepackage{xspace}

\newcommand{\magccdh}{\hyperref[Algorithm: Generalized CCDH Approx]{\texttt{ApproxCCDH}}}

\newcommand{\streamccdh}{\hyperref[Algorithm: Streaming CCDH Approx]{\ensuremath{\texttt{BasicStreamCCDH}}}}

\newcommand{\queryccdh}{\hyperref[Algorithm: Query CCDH Approx]{\ensuremath{\texttt{QueryCCDH}}}}

\newcommand{\updqueryccdh}{\hyperref[Algorithm: Query CCDH Approx Updated]{{\texttt{AdaptQueryCCDH}}}}

\newcommand{\actver}{\hyperref[Algorithm: Active Vertex Sampling]{\texttt{ActVerSamp}}}

\newcommand{\updccdh}{\hyperref[Algorithm: Generalized CCDH Approx Updated]{\texttt{AdaptCCDH}}}

\newcommand{\updstrm}{\hyperref[Algorithm: Streaming CCDH Approx Improved]{\ensuremath{\texttt{UpdStreamCCDH}}}}

\newcommand{\strmccdhalg}{\hyperref[Alg: Stream CCDH Final]{\ensuremath{\texttt{StreamCCDH}}}}
\newcommand{\field}[1]{\mathbb{#1}}
\newcommand{\R}{\field{R}}
\newcommand{\Nat}{\field{N}}

\newcommand{\func}[3]{{#1} : {#2} \rightarrow {#3}}

\DeclarePairedDelimiter{\floor}{\lfloor}{\rfloor}

\newcommand{\blue}[1]{\textcolor{blue}{#1}}

\newcommand{\sequence}[2]{{#1}_1,{#1}_2,...,{#1}_{#2}}
\newcommand{\fbrac}[1]{\left({#1}\right)}
\newcommand{\sbrac}[1]{\left\{{#1}\right\}}
\newcommand{\tbrac}[1]{\left[{#1}\right]}
\newcommand{\abs}[1]{\left|{#1}\right|}
\newcommand{\size}[1]{\left|{#1}\right|}
\newcommand{\constant}{c}

\DeclareMathOperator*{\Prob}{\field{P}}
\DeclareMathOperator*{\E}{\field{E}}
\DeclareMathOperator*{\Var}{\mathrm{Var}}
\DeclareMathOperator*{\Cov}{\mathrm{Cov}}

\newcommand{\uniform}{\mathrm{Unif}}

\newcommand{\approxerror}{\varepsilon}
\newcommand{\confidence}{\delta}

\newcommand{\bigo}[1]{O\fbrac{{#1}}}
\newcommand{\bigot}[1]{\widetilde{O}\fbrac{{#1}}}
\newcommand{\smallo}[1]{o\fbrac{{#1}}}

\newcommand{\bigomega}[1]{\Omega\fbrac{{#1}}}

\newcommand{\bigtheta}[1]{\Theta\fbrac{{#1}}}

\newcommand{\algo}{\sA}
\newcommand{\poly}[1]{\mathrm{poly}\fbrac{#1}}

\newcommand{\polylogeps}{\poly{\log\vertexcount,\approxerror^{-1}}}

\newcommand{\graph}{G}
\newcommand{\vertexset}{V}
\newcommand{\edgeset}{E}
\newcommand{\vertexcount}{n}
\newcommand{\edgecount}{m}
\newcommand{\vertex}{v}
\newcommand{\altvertex}{u}
\newcommand{\edge}{e}
\newcommand{\degree}[1]{d_{#1}}
\newcommand{\hindex}{h}

\newcommand{\randedgeq}{\texttt{RandEdge}}
\newcommand{\neighbourq}{\texttt{Nbr}}
\newcommand{\edgeexistsq}{\texttt{Pair}}
\newcommand{\degreeq}{\texttt{Deg}}

\newcommand{\domain}{\sD}
\newcommand{\ran}{\sR}

\newcommand{\ccspace}{Q}

\newcommand{\disjointness}[1]{\texttt{DISJ}_{#1}}
\newcommand{\promisedisjointness}[1]{\texttt{DISJ}^{prm}_{#1}}

\newcommand{\stringlength}{M}

\newcommand{\alicestring}{x}
\newcommand{\bobstring}{y}
\newcommand{\ccgraph}{{\graph_{\alicestring,\bobstring}}}
\newcommand{\ccvertexset}{\vertexset_{\alicestring,\bobstring}}
\newcommand{\ccedgeset}{\edgeset_{\alicestring,\bobstring}}

\newcommand{\embedding}{\sE}

\newcommand{\remove}[1]{}

\usetikzlibrary{backgrounds}

\usepackage{booktabs}

\begin{document}

\maketitle

\begin{abstract}

  Given a graph $G=(V,E)$ with $n$ vertices and $m$ edges, let $C(d)=|\{v\in V:\deg(v)\ge d\}|$ denote its complementary cumulative degree histogram (CCDH), and let $h$ be the $h$-index of $G$. The goal is to compute a $(\varepsilon_D,\varepsilon_R)$-bi-criteria multiplicative approximation of $C$, where $\varepsilon_D$ and $\varepsilon_R$ control the approximation in the domain and range, respectively. In this work, we study the complexity of this problem in the streaming model and an augmented query model that provides access to uniformly random edge queries in addition to standard local queries.

We design a model-agnostic algorithm that combines uniformly sampled vertices and edges to estimate the CCDH. Ignoring the dependence on the approximation parameters and polylogarithmic factors in $n$, the complexity in the streaming model is $\Theta(m/h)$: the upper bound holds for a single-pass algorithm, while the lower bound holds even for multi-pass algorithms. In the query model, the complexity is $\Theta(m/h)$ for adaptive algorithms. For non-adaptive algorithms, we give an upper bound of $O(n/h+m/h)$ and establish lower bounds through a reduction from estimating the number of distinct elements in a multiset when sampling without replacement. 

Our results resolve an open problem in sublinear graph algorithms posed in Sublinear.info (Problem 98). Finally, we validate our algorithm through extensive empirical evaluations on real-world graphs.


\end{abstract}

\vfill

\paragraph*{AI use statement}
AI tools were used for editorial purposes, to assist with the pictures for the lower bound, and empirical evaluations. All mathematical ideas and proofs presented in this work are the authors' own. The authors take full responsibility for the content of the submitted manuscript.

\newpage
\thispagestyle{empty}
\tableofcontents

\newpage
\setcounter{page}{1}

\section{Introduction}

The degree distribution of a graph is one of the fundamental objects in the study of massive graphs, and a foundational topic in network science~\cite{ BarbasiAlbert/Science/1999/EmergenceofScalinginRandomNetworks,FaloutsosCubed/SIGCOMMCCR/1999/OnPowerLawRelationShipsoftheInternetTopology,BroderKMRRSTW/CompNet/2000/GraphStructureinTheWeb,Newman/SiamRev/2003/StructureandFunctionofComplexNetworks,book/network-science}. Studying the degree distribution of graphs has been leveraged across a wide range of domains, including biology~\cite{PrzuljNatasa/Bioinformatics/2007/DegreeDistBiologicalNetComp}, social networks~\cite{Ebbes/WITS/2008/SamplingLargeScaleSocialNetworks,SeshadriKoldaPinar/PhysRevE/2012/CommunityStructureandDegDist,Golbeck/Book/2013/SocialWebAnalyzing/DegDist,kaufmann2025rumourspreadingdependslatent}, network robustness~\cite{YuanShaoStanleyHavlin/PhysRevE/2015/BreadthDegDistRobustness}, web graphs~\cite{CohenErezBenAvrahamHavlin/PhysRevLett/2000/WebGraphDegreeDist,PennockFlakeLawrence/PNAS/2002/WebDegDist,AksoyKoldaPinar/JCompNet/2017/MeasuringandModelingBipartiteGraphs}, graph mining~\cite{ChakrabartiFaloutsos/ACMCSurvey/2006/GraphMiningSurvey}, study of algorithms~\cite{Mitzenmacher/SurveyBook/2004/BriefSurveyPowerLaw,DuarkKoldaPinarSeshadri/NSW/2013/NullModelforAllDegreeDistribution}, and analysis of complex networks~\cite{AliakbaryHM/TCJ/2025/DegreeDistributionCompNetAnalysis,FranceschiPZ/PRS-A/2025/DegDistLargeNonGrowingNetworks}. More recently, the degree distribution has also been used in the study of graph generative models~\cite{AlamP/FrontBigData/2021/GenerationLargeGraphDegDist,ChenHHL/ICML/2023/GraphGenerationDegDist,AlabiK/KDD/2025/SeparatingGGMDegDist}. These works highlight the degree distribution as an important summary of graph structure.


Given full access to a graph in the standard RAM model, the degree distribution can be computed exactly in linear time and space. However, such access models, in which the entire edge set is known a priori are increasingly unrealistic for modern graphs, which are often massive, dynamic, or distributed across multiple systems. Consequently, there has been sustained interest in designing graph algorithms under alternative access models. In practice, graphs may be accessed through sampling mechanisms~\cite{LeskovecFaloutsos/KDD/2006/SamplingFromLargeGraphs,Ebbes/WITS/2008/SamplingLargeScaleSocialNetworks}, stored and processed in a distributed fashion~\cite{book/biggraphdistributed,ICMD2020/BuragohainRisvik}, as massive streams~\cite{Mcgregor/ACMSigmod/GraphStreamAlgorithms,IEEEKDE/GouZouZhaoYang}, or through restricted queries to the graph~\cite{GoldreichRon/STOC/97/PropertyTestinginBoundedDegreeGraphs,Goldreich/Book/2010/GraphPropertyTestingSurvey}. 

There is a vast literature on estimating the degree distribution using a small subsample~\cite{LeskovecFaloutsos/KDD/2006/SamplingFromLargeGraphs, RibeiroTowsley/CDC/2012/DegreeDistributionGraphSampling,www/dasguptaKumarSarlos,ZhangKolaczykSpencer/AnnalsAS/2015/EstimatingDegreeDistributions,Antunes02102021,STEWART2025105420}. However, most of these approaches either rely on asymptotic arguments, heuristic methods, or do not provide non-asymptotic performance guarantees. In the query model, Eden et al.~\cite{Eden/WWW/2017/CCDHinQueryModel} proposed an algorithm with theoretical guarantees. In the streaming model, Simpson et al.~\cite{SeshadriMcGregor/ICDM/2015/CCDHStreamingEmpirical} proposed a heuristic algorithm and evaluated it empirically. This motivates our first question:
\begin{center}
    \emph{\hypertarget{Q1}{\textbf{Q1.}} Can we design a unified, model-agnostic, theoretically grounded algorithm for estimating the degree distribution across streaming, query, and other sampling-based access models?
}
\end{center}

On the other hand, the fundamental complexity of estimating the degree distribution under restricted access remains unknown. No lower bound is known for the problem in any of the models discussed above. For the query-based access model, the problem of establishing a non-trivial lower bound has been explicitly posed as an open problem in  sublinear.info~\cite{sublinearProblemEstimating}. This motivates our second question:

\begin{center}
    \emph{\hypertarget{Q2}{\textbf{Q2.}} Can we establish a lower bound for the problem of estimating the degree distribution?}
\end{center}

Together, these questions aim to characterize both the algorithmic possibilities and the fundamental limits of degree distribution estimation in various graph access models. In this work, we answer both questions in the affirmative.


\noindent\textbf{Unified Estimation Framework(\hyperlink{Q1}{Q1}):} We design a model-agnostic algorithm that combines vertex and edge samples effectively to establish a theoretically valid estimator of the degree distribution. This algorithm can be realized in any access model where such samples can be obtained. We also outline how to realize the algorithm in the query and streaming models. To the best of our knowledge, ours is the first theoretically guaranteed estimator of the degree distribution in the streaming model. 

\noindent\textbf{Lower Bounds(\hyperlink{Q2}{Q2}):} We establish lower bounds on both the streaming and query-based access models. For both models, our lower bounds are tight on the graph parameters. Our results for the query model also imply a lower bound on the number of samples required for any algorithm that uses vertex and edge samples to estimate the degree distribution. As in the case of upper bounds, the lower bounds for multi-pass streaming model, and adaptive query model also rely on a unified framework across the models. For non-adaptive query algorithms, the lower bound follows from a eduction from estimating the number of distinct elements in a multiset when sampling without replacement.



 \noindent\textbf{Organization of the Paper:} We describe the problem in \Cref{Subsection: Problem Details}, introduce the sublinear models in \Cref{Subsection: Problem Setup - Streaming Model}, and outline our results in \Cref{ssec:ourresults}. The algorithms are presented in \Cref{Section: Upper Bounds}, with some of the proofs presented in \Cref{App: Proofs for MAgAlgo,AppSec: Streaming Implementation Proofs,AppSubsec: Query Algo Implementation}. The unified lower bound for multi-pass streaming and adaptive query model is presented in \Cref{Section: Lower Bounds}, with some of the proofs deferred to \Cref{App: Proofs Lower Bounds}. The lower bound for the non-adaptive query model is presented in \Cref{SubSec: Non_Adaptive LB}. Finally, \Cref{AppSec: Experiments} contains the experimental results.
 

\section{Definition, Models, and Results}\label{Section: Prelims}

In this section, we outline the problem formulation (\Cref{Subsection: Problem Details}), and describe the computing models (\Cref{Subsection: Problem Setup - Streaming Model}), and discuss our results (\Cref{ssec:ourresults}).


\subsection{Problem Definition}\label{Subsection: Problem Details}

We consider simple graphs $\graph = (\vertexset,\edgeset)$, where $\size{\vertexset} = \vertexcount$ and $\size{\edgeset} = \edgecount$. 
For any vertex $\vertex \in \vertexset$, let $\degree{\vertex}$ denote its degree. 
We denote the maximum degree of the graph by $\degree{\max{}} = \max_{\vertex \in \vertexset}\degree{\vertex}$. We denote the set of active vertices, i.e. vertices with at least one edge incident to them as $\actverset = \sbrac{\vertex\in\vertexset \mid \degree{\vertex} \geq 1}$,  and $\actvercount = |\actverset|$. Observe that for a simple graph, we have $2\edgecount \ge \actvercount$.  Throughout this work, we assume the set of vertices $\vertexset$ is of the form $[\vertexcount]$, where $[p]$ denotes the set of natural numbers $\sbrac{1,2,\ldots,p}$. The \emph{h-index} of the graph, denoted by $\hindex$, is the largest integer $d$ such that the graph contains at least $d$ vertices of degree at least $d$. 


The degree distribution of a graph with $\vertexcount$ vertices is the sequence $\vertexcount_{0}, \vertexcount_{1}, \ldots, \vertexcount_{\vertexcount}$, where $\vertexcount_{d}$ denotes the number of vertices with degree $d$. Motivated by the findings of Clauset et al.~\cite{siamrev/ClausetSN09}, prior works~\cite{RibeiroTowsley/CDC/2012/DegreeDistributionGraphSampling,SeshadriMcGregor/ICDM/2015/CCDHStreamingEmpirical,Eden/WWW/2017/CCDHinQueryModel} adopt the notion of the \emph{Complementary Cumulative Degree Histogram} (ccdh) in place of the degree histogram, as the ccdh is more robust to noise and monotonically decreasing. The ccdh, defined below, is based on the number of vertices whose degree is at least a given threshold.

\begin{definition}[CCDH]\label{Definition: CCDH}
    Given a graph $\graph$, let $\func{\ccdh}{\tbrac{\vertexcount}}{\tbrac{\vertexcount}}$ be the {Complementary Cumulative Degree Histogram} (ccdh),  defined as $\ccdh(d) \triangleq \sum_{r \geq d} \vertexcount_r = \size{\sbrac{\vertex \in \vertexset~|~\degree{\vertex}\geq\degree{}}}$.
\end{definition}

Standard distance measures such as $\ell_p$ or total variation are often inadequate for comparing degree distributions of graphs, as they allow small-mass regions of the distribution to be ignored entirely. In particular, such regions may correspond to high-degree vertices that play a critical structural role in the graph. To address this, prior works~\cite{SeshadriMcGregor/ICDM/2015/CCDHStreamingEmpirical,Eden/WWW/2017/CCDHinQueryModel} introduced a bi-criteria multiplicative approximation notion that simultaneously allows controlled multiplicative distortion in both the domain and the range of the complementary cumulative degree histogram (ccdh).

\begin{definition}[$(\approxdom,\approxran)$-BiCriteria Multiplicative Approximation]\label{Definition: BiCriteria Multiplicative Approximation}
Given a monotonically decreasing function $\func{f}{\domain}{\ran}$, where $\domain,\ran \subseteq \R$, and parameters $\approxdom,\approxran \in (0,1)$, an approximation $\hat{f}$ is said to be a $(\approxdom,\approxran)$-bi-criteria multiplicative approximation of $f$ if $\forall x \in \domain$,
\[
(1-\approxran)f((1+\approxdom)x) \leq \hat{f}(x) \leq (1+\approxran)f((1-\approxdom)x)
\]
\end{definition}

This criterion is weaker than point-wise multiplicative approximation, yet stronger than standard distributional distances. 
In the rest of this work, we refer to this criterion as $\biapprox$-BMA. Our objective is to obtain a $\biapprox$-BMA of the ccdh of the graph.



\subsection{Models: Streaming and Query}\label{Subsection: Problem Setup - Streaming Model}

\subparagraph*{Streaming Model:} In the streaming model, we consider the edge-arrival model of graph streaming. In the insertion-only variant of the model, we are presented with a stream of edges $\sbrac{\edge_i}_{i \in [\edgecount]}$ of a graph $G$, where each edge arrives in an arbitrary order. In the turnstile variant of the model, deletion of edges over the stream is allowed as well as insertion. The models are also categorized as one-pass or multi-pass,~\footnote{By multi-pass, we mean that the number of passes is bounded by some constant.} depending on how many passes are allowed over the stream.



\subparagraph*{Query Model:} For the query model, we assume access to the graph through the following query primitives. Given a vertex $\vertex$, \emph{degree query} $\degreeq(\vertex)$ returns its degree $\degree{\vertex}$. A \emph{random edge query} $\randedgeq$ returns an edge $\edge \in \edgeset$ chosen uniformly at random. Given a vertex $\vertex$ and an index $i \in \tbrac{\vertexcount}$, \emph{neighbor query} $\neighbourq(\vertex,i)$, returns the $i$-th neighbor of $\vertex$ if such a neighbor exists, and $\perp$ otherwise. Given two vertices $\altvertex,\vertex$, \emph{pair query} $\edgeexistsq(\altvertex,\vertex)$, returns $1$ if $(\altvertex,\vertex) \in \edgeset$, and $0$ otherwise.
The \degreeq{}, \neighbourq{}, and \edgeexistsq{} queries, collectively known as local queries are well-known in the query model~\cite{Goldreich/Book/2010/GraphPropertyTestingSurvey}. The \randedgeq{} query has been considered in more recent works~\cite{Aliakbarpour/Algorithmica/2018/StarSubgraphEdgeSampling,assadi2018simple,EdenRosenbaum/Approx/2018/LowerBoundGraphCommunication,TetekThorup/STOC/2022/EdgeSamplingFullNbrhood,BerettaTetek/TALG/2024/BetterSumEstimationViaWeightedSampling,EdenRV/ArXiV/2025/TestableAlgoTrianglEdges,BerettaCS/ArXiv/2025/FasterEdgeEstimation} and is similar to the edge sampling framework for graph sampling \cite{StumphWiuf/PhyiscalReview/2005/SamplingGraphsDegreeDist,LeskovecFaloutsos/KDD/2006/SamplingFromLargeGraphs}. We consider both non-adaptive and adaptive models of graph property testing~\cite{GoldreichWigderson/FOCS/2022/NonAdaptivevsAdaptiveDenseGraphTesting}. In the non-adaptive model, all queries are fixed in advance as a function of the input parameters, whereas in the adaptive model, each query may depend on the outcomes of previous queries. 




\subsection{Our Results}
\label{ssec:ourresults}

In this section, we give an overview of the results presented in this work. \Cref{Tab: Results Table} sums up our results. We present the problem definition, the model, and algorithms for graph data. However, the results also extend to more general settings. We discuss this in detail in \Cref{Remark: Non-Graph Extension,Remark: Lower Bound Generalization}. Our  results establish optimal bounds for estimating the degree distribution, up to $\polylogeps$ factors, for all the models considered in this work.

\begin{table*}[ht!]
\centering
\small
\setlength{\tabcolsep}{6pt}
\renewcommand{\arraystretch}{1.2}
\begin{tabularx}{\textwidth}{@{} l X X @{}}
\toprule
& \textbf{Streaming Model} & \textbf{Query Model} \\
\midrule

\multirow[c]{2}{*}{\textbf{Upper Bound}}
& \textbf{Single-Pass}
  (\Cref{Thm: Final Stream CCDH}) \newline
  $\bigot{\min\!\left\{
    \vertexcount,\,
    \frac{\edgecount}{\approxran^{2}\hindex}
    + \frac{\edgecount}{\approxdom^{2}\hindex}
  \right\}}$
& \textbf{Non-Adaptive}
  (\Cref{Theorem: Query CCDH Algorithm Guarantee}) \newline
  $\bigot{\min\!\left\{
    \vertexcount,\,
    \frac{\vertexcount}{\approxran^{2}\hindex}
    + \frac{\edgecount}{\approxdom^{2}\hindex}
  \right\}}$
\\[8pt]

&
& \textbf{Adaptive}
  (\Cref{Theorem: Query Guarantee Updated}) \newline
  $\bigot{\min\!\left\{
    \vertexcount,\,
    \frac{\edgecount}{\approxran^{2}\hindex}
    + \frac{\edgecount}{\approxdom^{2}\hindex}
  \right\}}$
\\[8pt]

\midrule

\multirow[c]{2}{*}{\textbf{Lower Bound}}
& \textbf{Multi-Pass}
  (\Cref{Theorem: Streaming General LB,Theorem: Streaming h-index Lower Bound}) \newline
  $\bigomega{\min\!\left\{
    \vertexcount,\,
    \frac{\edgecount}{\hindex}
  \right\}}$
& \textbf{Non-Adaptive}
  (\Cref{Theorem: Query General LB,Theorem: Query h-index Lower Bound,Theorem: Query Non-adaptive Lower Bound}) \newline
  $\bigomega{\max\sbrac{\min\!\left\{
    \vertexcount,\,
    \frac{\edgecount}{\hindex}
  \right\},\min\!\left\{
    \frac{\vertexcount}{\hindex},\,
    \frac{\edgecount}{\log \hindex}
  \right\}}}$
\\[8pt]

&
& \textbf{Adaptive}
  (\Cref{Theorem: Query General LB,Theorem: Query h-index Lower Bound}) \newline
  $\bigomega{\min\!\left\{
    \vertexcount,\,
    \frac{\edgecount}{\hindex}
  \right\}}$
\\[8pt]

\bottomrule
\end{tabularx}
\caption{Overview of the upper and lower bounds established in our work.}
\label{Tab: Results Table}
\end{table*}








\subparagraph*{Upper Bounds.}
We present a model-agnostic algorithm for estimating the ccdh that combines uniformly sampled vertices and edges. 
The algorithm outputs an estimate that is a $\biapprox$-BMA of the true ccdh using 
$\widetilde{O}\!\left(\nicefrac{\vertexcount}{\approxran^2 \hindex}\right)$\footnote{The $\bigot{\bullet}$ notation hides $\mathrm{poly}(\log\vertexcount)$ factors.} uniformly sampled vertices
and $\widetilde{O}\!\left(\nicefrac{\edgecount}{\approxdom^2 \hindex}\right)$ uniformly sampled edges. To instantiate this generic procedure under a specific access model, it suffices to implement the required vertex and edge sampling primitives. 



\noindent\textbf{Query Model:} Our algorithms in the query model uses only \degreeq{} and \randedgeq{} queries. In the non-adaptive query model, the model-agnostic algorithm can be implemented using $\widetilde{O}\!\left(\nicefrac{\vertexcount}{\approxran^2 \hindex}
+
\nicefrac{\edgecount}{\approxdom^2 \hindex}\right)
$
queries. For the adaptive query model, we show that the algorithm can be implemented using $\widetilde{O}\!\left(\nicefrac{\edgecount}{\approxran^2 \hindex} + \nicefrac{\edgecount}{\approxdom^2 \hindex}\right)$ queries. In particular, this yields asymptotically smaller complexity than the non-adaptive and one-pass models when $\edgecount = o(\vertexcount)$.

\noindent\textbf{Streaming Model:} We show that the model-agnostic algorithm can be implemented in a one-pass streaming model using $\widetilde{O}\!\left(\nicefrac{\edgecount}{\approxran^2 \hindex} + \nicefrac{\edgecount}{\approxdom^2 \hindex}\right)$ space. This algorithm can also be extended to the turnstile variant of the streaming model, as outlined in \Cref{Remark: Turnstile Extension}.

\subparagraph*{Lower Bounds.}

For our lower bounds, we first consider the insertion-only multi-pass streaming model and the adaptive query model with access to $\degreeq$, $\neighbourq$, $\edgeexistsq$, and $\randedgeq$ queries. 
We show that for any $\vertexcount \in \Nat$, there are graphs with $\vertexcount$ vertices and $\bigtheta{\vertexcount}$ edges for which any such algorithm that computes an approximate ccdh requires $\bigomega{\vertexcount}$ queries or space, respectively. This rules out the possibility of sublinear algorithms for estimating the ccdh on arbitrary graphs, even with multiple passes or adaptivity.


Next, we establish a parameterized lower bound: for any $\vertexcount, \edgecount, \hindex \in \Nat$, any query or streaming algorithm that computes an approximate ccdh of a graph with $\vertexcount$ vertices, $\edgecount$ edges, and h-index $\hindex$ must use at least $\bigomega{\nicefrac{\edgecount}{\hindex}}$ queries or space. This shows that our algorithms for the streaming model, and the adaptive query model are optimal for all graphs. This establishes $\nicefrac{\edgecount}{\hindex}$ as an inherent measure of the difficulty of estimating ccdh in these models. This also establishes that our algorithm for the non-adaptive query model is optimal for all graphs with $\edgecount = \bigomega{\vertexcount}$.

Finally, we show that in the non-adaptive query model, for any $\vertexcount, \edgecount, \hindex \in \Nat$ with $\edgecount \leq \nicefrac{\vertexcount}{2}$, any algorithm that computes an approximate ccdh of a graph with $\vertexcount$ vertices, $\edgecount$ edges, and h-index $\hindex$ must use at least $\bigomega{\min\sbrac{\nicefrac{\edgecount}{\log \hindex},\nicefrac{\vertexcount}{\hindex}}}$ queries. This establishes the optimality of our non-adaptive query algorithm.


\subparagraph*{Empirical Evaluations}

Due to the practical nature of the problem, we empirically evaluate the performance of our estimator on real-world graphs, and validate its performance against the existing baseline algorithms. We also study the complexity of estimating the ccdh across various real-world datasets. The details are presented in \Cref{AppSec: Experiments}.


\section{Upper Bounds}\label{Section: Upper Bounds}

In this section, we describe the estimation algorithms. We describe our model-agnostic estimator in \Cref{SubSec: MAgAlgo}. Implementations in streaming and query models are presented in \Cref{SubSec: Implementation in Streaming} and \Cref{Subsection: Query Algorithm}, respectively.

\subsection{Model-Agnostic Algorithm}\label{SubSec: MAgAlgo}

Our algorithm approximates the complementary cumulative degree histogram (ccdh) by decomposing it into two regions: a \emph{head}, corresponding to small degrees with large ccdh values, and a \emph{tail}, corresponding to large degrees containing a small number of high-degree vertices. We design separate estimators for the head and tail regions and develop a principled h-index–based mechanism to combine these two regimes into a single estimator.

Each degree value in the head region is supported by many vertices, allowing the algorithm to estimate this region using the empirical distribution of uniformly sampled vertex degrees. In contrast, vertices in the tail region have large degrees; the algorithm therefore samples edges uniformly and uses the resulting degree estimates of high-degree vertices to construct the ccdh in this regime.

We start with an algorithm that is parameterized by two thresholds: a ccdh (value) threshold $\ccdhthreshold$ and a degree (domain) threshold $\degreethreshold$. 
The head region consists of all degrees $d$ such that $\ccdh(d) \ge \ccdhthreshold$, while the tail region consists of all degrees $d \ge \degreethreshold$. We finally combine these thresholds through the h-index $\hindex$ of the graph. We now state and analyze the algorithm. The proofs are in~\Cref{App: Proofs for MAgAlgo}.

\begin{algorithm}[ht!]
    \caption{\magccdh}\label{Algorithm: Generalized CCDH Approx}
    \begin{algorithmic}[1]
        \REQUIRE $\vertexcount,\edgecount,\degreethreshold,\ccdhthreshold,\approxdom,\approxran$, where $\ccdh\fbrac{\degree{}} \geq \ccdhthreshold, \forall \degree{} \leq \degreethreshold$
        \STATE $\vertexsample \gets$ Sample $\vertexsamplesize = \frac{\constant\vertexcount}{\ccdhthreshold\approxran^2}\log\vertexcount$ vertices u.a.r.\label{Line: General Algo Vertex Sample Size}
        \STATE $\degreesample \gets $ Degree of vertices in $\vertexsample$
        \STATE $\edgesample \gets$ Sample $\edgesamplesize = \frac{\constant\edgecount}{\degreethreshold\approxdom^2}\log\vertexcount$ edges u.a.r. \label{Line: General Algo Edge Sample Size}
        \STATE $\approxdegree{\vertex} = \frac{\edgecount\size{\sbrac{\edge \in \edgesample|\vertex \in \edge}}}{\edgesamplesize}, \forall \vertex \in \vertexset$\label{Line: General Algorithm Approx Degree} 
        \IF{$\degree{} \leq \degreethreshold$}
            \STATE $\approxccdh(\degree{}) \gets \frac{\vertexcount\size{\sbrac{\vertex \in \vertexsample|\degree{\vertex} \geq \degree{}}}}{\vertexsamplesize} =\frac{\vertexcount\size{\sbrac{\degree{i} \in \degreesample|\degree{i} \geq \degree{}}}}{\vertexsamplesize}$\label{Line: General Algorithm CCDH Based Approx}
        \ENDIF
        \IF{$\degree{} > \degreethreshold$}
            \STATE $\approxccdh(\degree{}) \gets \size{\sbrac{\vertex|\approxdegree{\vertex}\geq \degree{}}}$\label{Line: General Algorithm Degree Based Approx}
        \ENDIF
    \end{algorithmic}
\end{algorithm}

We first provide the approximation guarantees of the estimators for the head and tail regions.

\begin{restatable}[Head Estimator Guarantee]{lemma}{headest}\label{Lemma: General Algorithm High CCDH degree Guarantee}
    For all $\degree{}$ such that $\ccdh(\degree{}) \geq \ccdhthreshold$, \magccdh{} returns $\approxccdh(\degree{})$ such that $\approxccdh(\degree{}) \in \tbrac{\fbrac{1-\approxran}\ccdh(\degree{}),\fbrac{1+\approxran}\ccdh(\degree{})}$ with high probability.
\end{restatable}

\begin{restatable}[Tail Estimator Guarantee]{lemma}{tailest}\label{Lemma: General Algorithm High Degree Vertex Guarantee}
    For all $\degree{} > \degreethreshold$, \magccdh{} returns $\approxccdh(\degree{})$ such that $\approxccdh(\degree{}) \in [\ccdh((1+\approxdom)\degree{}),\ccdh\fbrac{\fbrac{1-\approxdom}\degree{}}]$ with high probability.
\end{restatable}

Now, we combine the approximation guarantees for the cases when $\ccdh(\degree{}) \geq \ccdhthreshold$ (Lemma~\ref{Lemma: General Algorithm High CCDH degree Guarantee}) and the case when $\degree{} > \degreethreshold$ (Lemma~\ref{Lemma: General Algorithm High Degree Vertex Guarantee}) to obtain an estimation guarantee of \magccdh{} algorithm for all values of degree $\degree{}$. 

\begin{restatable}[\magccdh{} Guarantee]{theorem}{magalgo}\label{Theorem: General CCDH Algorithm Guarantee}
    For a graph $\graph$ with h-index $\hindex$, given an $\approxhindex \leq \hindex$ \magccdh{} uses $\bigo{\nicefrac{\vertexcount}{\approxhindex\approxran^2}\log\vertexcount}$ degree samples, and $\bigo{\nicefrac{\edgecount}{\approxhindex\approxdom^2}\log\vertexcount}$ edge samples and outputs $\approxccdh$, an $\biapprox$-BMA of the original $\ccdh$ with high probability.
\end{restatable}

\begin{proof}
    Given $\approxhindex \leq \hindex$, let \magccdh{} decide $\ccdh\fbrac{\degree{}} \geq \ccdhthreshold$ if $\degree{} \leq \approxhindex$, and $\degree{} > \degreethreshold$ if $\degree{} > \approxhindex$. Also, let the algorithm obtain $(\nicefrac{\constant \vertexcount}{\approxhindex\approxran^2})\log\vertexcount$ vertex samples, and $(\nicefrac{\constant \edgecount}{\approxhindex\approxdom^2})\log\vertexcount$ edge samples. 
    By the fact that $\approxhindex \leq \hindex$, monotonicity of ccdh, and the definition of $\hindex$-index, for all degrees $\degree{} \leq \approxhindex$, we have:
    \begin{align*}
    \ccdh\fbrac{\degree{}} \geq \ccdh\fbrac{\approxhindex} \geq \ccdh\fbrac{\hindex} \geq \hindex
    \end{align*}
    Then, by~\Cref{Lemma: General Algorithm High CCDH degree Guarantee}, given $\fbrac{\nicefrac{\vertexcount}{\approxhindex\approxran^2}}\log\vertexcount \geq \fbrac{\nicefrac{\vertexcount}{\hindex\approxran^2}}\log\vertexcount$ vertex samples, we have with high probability for all degree $\degree{} \leq \approxhindex$:
    \begin{align}
        \approxccdh(\degree{}) \in \tbrac{\fbrac{1-\approxran}\ccdh(\degree{}),\fbrac{1+\approxran}\ccdh(\degree{})}\label{Equation: General Algorithm Head Capture}
    \end{align}
    Also, for all degrees $\degree{} > \approxhindex$, given $(\nicefrac{\constant \edgecount}{\approxhindex\approxdom^2})\log\vertexcount$ samples, by Lemma~\ref{Lemma: General Algorithm High Degree Vertex Guarantee}, we have with high probability:
    \begin{align}
        \approxccdh(\degree{}) \in \tbrac{\ccdh\fbrac{\fbrac{1+\approxdom}\degree{}},\ccdh\fbrac{\fbrac{1-\approxdom}\degree{}}}\label{Equation: General Algorithm Tail Capture}
    \end{align}
    Combining Equations~\ref{Equation: General Algorithm Head Capture} and~\ref{Equation: General Algorithm Tail Capture}, we have that \Cref{Algorithm: Generalized CCDH Approx} ensures $\approxccdh$ is a $\biapprox$-BMA of $\ccdh$ for all values of $\degree{}$ with high probability.
\end{proof}

We can combine a result due to Assadi and Nguyen~\cite{AssadiNguyen/Approx/2022/OptimalH-IndexAndTriangle} with~\Cref{Theorem: General CCDH Algorithm Guarantee} to obtain an algorithm that does not need prior knowledge of the h-index.

\begin{restatable}[Model-Agnostic Estimator]{theorem}{magalgoplus}\label{Theorem: General CCDH Algorithm Guarantee no hindex}
    For a graph $\graph$ with h-index $\hindex$, there exists an algorithm that uses $\bigo{\nicefrac{\vertexcount}{\approxhindex\approxran^2}\log\vertexcount}$ degree samples, and $\bigo{\nicefrac{\edgecount}{\approxhindex\approxdom^2}\log\vertexcount}$ edge samples and its output $\approxccdh$ is an $\biapprox$-BMA of the original $\ccdh$ with high probability.
\end{restatable}

\begin{remark}[{Extension for non-graph data}]\label{Remark: Non-Graph Extension}
  
Note that both the problems of computing ccdh and the notion of h-index can also be formulated for non-graph data, such as multisets over a universe. Our algorithms do not use any structural property of the underlying graph. Hence, our upper bounds extend to settings where we are given query or streaming access to a multiset. In the case of a multiset, a degree sample corresponds to sampling an element uniformly at random and querying its frequency, while an edge sample corresponds to sampling an element with probability proportional to its frequency. This framework is closely connected to the weighted sampling model \cite{motwani07,BerettaTetek/TALG/2024/BetterSumEstimationViaWeightedSampling}, and to evaluation oracle models in distribution testing \cite{Batu/CCC/2002/DualModelEntropy,RubinfeldServedio/RSA/2009/TestingMonotoneHDD,Canonne/ICALP/2014/DualModelTesting}.
  
\end{remark}

\subsection{Implementation in the Streaming Model}\label{SubSec: Implementation in Streaming}

We now describe the implementations of the \magccdh{} algorithm in the streaming  model. We start with a straightforward implementation of \magccdh{} by obtaining uniformly random samples of vertices and edges, with replacement. For vertex samples, we select $\vertexsamplesize$ vertices uniformly at random with replacement from $\vertexset$ and maintain exact degree counts for these vertices throughout the stream. For edge samples, we initialize $\edgesamplesize$ independent reservoir samplers of size one~\cite{Vitter/TOMS/1985/ReservoirSampling} at the start of the stream and update them as the stream evolves.

\begin{algorithm}[ht!]
    \caption{\streamccdh}\label{Algorithm: Streaming CCDH Approx}
    \begin{algorithmic}[1]
        \REQUIRE $\vertexcount,\edgecount,\approxhindex \leq \hindex,\approxdom,\approxran$
        \FUNCTION{Initialize}{}
            \STATE $\vertexsamplesize \gets \frac{\constant\vertexcount}{\approxhindex\approxran^2}\log\vertexcount$
            \STATE $\edgesamplesize \gets \frac{\constant\edgecount}{\approxhindex\approxdom^2}\log\vertexcount$
            \STATE $\vertexsample \gets \vertexsamplesize$ i.i.d. samples from $\uniform\fbrac{[\vertexcount]}$
            \STATE $\degreesample \gets$ Zero-Array of length $\vertexsamplesize$ 
            \STATE $\edgesample \gets \edgesamplesize$ reservoir samplers of size $1$
        \ENDFUNCTION
        \FUNCTION{Update($\edge = \fbrac{\altvertex,\vertex}$)}
            \FOR{$i \in \vertexsamplesize$}
                \IF{$\vertexsample[i] = \altvertex$ or $\vertexsample[i] = \vertex$}
                    \STATE $\degreesample[i] \gets \degreesample[i]+1$
                \ENDIF
            \ENDFOR
            \STATE Update reservoir samplers $\edgesample$
        \ENDFUNCTION
        \FUNCTION{Estimate}{}
            \STATE $\approxdegree{\vertex} = \frac{\edgecount\size{\sbrac{\edge \in \edgesample|\vertex \in \edge}}}{\edgesamplesize}, \forall \vertex \in \vertexset$
            \IF{$\degree{} \leq \approxhindex$}
                \STATE $\approxccdh(\degree{}) \gets \frac{\vertexcount\size{\sbrac{\degree{i} \in \degreesample|\degree{i} \geq \degree{}}}}{\vertexsamplesize}$
            \ENDIF
            \IF{$\degree{} > \approxhindex$}
                \STATE $\approxccdh(\degree{}) \gets \size{\sbrac{\vertex|\approxdegree{\vertex}\geq \degree{}}}$
            \ENDIF
        \ENDFUNCTION
    \end{algorithmic}
\end{algorithm}

\begin{restatable}[\streamccdh{} Guarantee]{theorem}{strmccdhthm}\label{Theorem: Streaming CCDH Algorithm Guarantee}
    \streamccdh{} is a one-pass streaming algorithm that uses $\bigot{\nicefrac{\vertexcount}{\approxhindex\approxran^2}+\nicefrac{\edgecount}{\approxhindex\approxdom^2}}$ space, and its output $\approxccdh$ is an $\biapprox$-BMA of the original $\ccdh$ with high probability.
\end{restatable}

\begin{proof}
    The algorithm uses $\bigo{(\nicefrac{\edgecount}{\approxhindex\approxdom^2})\log\vertexcount}$ reservoir samplers each of which uses $\bigot{1}$ space. Additionally, it uses $\bigo{(\nicefrac{\vertexcount}{\approxhindex\approxran^2})\log\vertexcount}$ counters to store the degrees of the selected vertices. 

    The algorithm ensures that $\vertexsample$ and $\edgesample$ contain $(\nicefrac{\vertexcount}{\approxhindex\approxran^2})\log\vertexcount$ and $(\nicefrac{\edgecount}{\approxhindex\approxdom^2})\log\vertexcount$ i.i.d. samples of degrees and edges, respectively. Hence, the guarantee on the output follows directly from the guarantee on \magccdh{} (\Cref{Theorem: General CCDH Algorithm Guarantee}).
\end{proof}

Next, we present an improved algorithm for sparse graphs, i.e. when $\edgecount = \smallo{\vertexcount}$. Our objective is to design an algorithm whose space complexity depends on $\edgecount$ rather than $\vertexcount$. Observe that the tail estimator, i.e. the estimator for degree values $\degree{}$ such that $\degree{} \geq \approxhindex$, uses $\bigot{\nicefrac{\edgecount}{\approxhindex}}$ space. Hence, we focus on the head part, i.e. the case when degree values $\degree{}$ are smaller than $\approxhindex$. 

Note that the isolated vertices do not contribute any information to the ccdh. Hence, it suffices for us to sample from the set of active vertices $\actverset$, with the space complexity depending on $\actvercount$ rather than $\vertexcount$. For this purpose, we use a ${2}$-wise independent hash functions to sample vertices to maintain degree counters. If a vertex is not in $\actverset$, it does not arrive in the stream and hence, does not require a degree counter. Therefore, the space required to maintain the degree counters, and correspondingly the head estimate depends on the number of active vertices $\actvercount$ as well as the hash function uses. We present our updated head estimator below in \updstrm{}.

\begin{algorithm}[ht!]
    \caption{\updstrm{}}\label{Algorithm: Streaming CCDH Approx Improved}
    \begin{algorithmic}[1]
        \REQUIRE $\vertexcount,\edgecount,\approxhindex \leq \hindex,\approxdom,\approxran$
        \FUNCTION{Initialize}{}
        \STATE $\hashrange \gets \floor{\constant\approxerror^2\approxhindex}$
        \STATE $\medcount \gets \constant\log\vertexcount$
            \STATE For each $j \in [\medcount]$, $\func{\hash_j}{[\vertexcount]}{[\hashrange]} \gets$ A ${2}$-wise independent hash function.
            \STATE $\degreesample_j \gets \emptyset \forall j \in [\medcount]$
        \ENDFUNCTION
        \FUNCTION{Update($\edge = \fbrac{\altvertex,\vertex}$)}
            \FOR{$j \in [{\medcount}]$}
                \FOR{$i \in \sbrac{\altvertex,\vertex}$}
                    \IF{$\hash_j(i) = 1$}
                        \IF{$ i \in \degreesample_j$}
                            \STATE Increase its degree count by $1$
                        \ELSE
                            \STATE Include $i$ in $\degreesample_j$ with degree count $1$
                        \ENDIF
                    \ENDIF
                \ENDFOR
            \ENDFOR
        \ENDFUNCTION
        \FUNCTION{Estimate}{}
            \IF{$\degree{} \leq \approxhindex$}
                \FOR{$j \in [{\medcount}]$}
                    \STATE $\approxccdh_j(\degree{}) \gets  \hashrange \size{\sbrac{\degree{i} \in \degreesample|\degree{i} \geq \degree{}}}$
                \ENDFOR 
                \STATE $\approxccdh(\degree{}) \gets \text{ median of } \sbrac{\approxccdh_i(\degree{})|i \in [\medcount]}$
            \ENDIF
        \ENDFUNCTION
    \end{algorithmic}
\end{algorithm}

We can combine the updated head estimator of \updstrm{} with the tail estimator of \streamccdh{} to obtain the following result. The complete algorithm \strmccdhalg{} and the proofs are presented in \Cref{AppSec: Streaming Implementation Proofs}.

\begin{restatable}[]{theorem}{strmccdhfinal}\label{Thm: Final Stream CCDH}
    \strmccdhalg{} is a one-pass streaming algorithm that uses $\bigot{\nicefrac{\edgecount}{\hindex}\fbrac{\nicefrac{1}{\approxdom^2}+\nicefrac{1}{\approxran^2}}}$ space and its output $\approxccdh$ is a $(\approxdom,\approxran)$-BMA of the original $\ccdh$ with high probability.
\end{restatable}

\begin{remark}[Extending to the Turnstile Model]\label{Remark: Turnstile Extension}
    Observe that this algorithm extends under standard turnstile model. For the degree counts, in case of deletion of an edge, we decrement the counters for its endpoints, if present. For obtaining the edge samples, it suffices to use an $\ell_0$ sampler to obtain a uniformly random sample from the edge set. Exact $\ell_0$ samplers that succeeds with high probability can be implemented using $\bigot{1}$ space~\cite{JowhariSauglamTardos/PODS/2011/TightBoundsLpSamplers}. 
\end{remark}

\subsection{Implementation in the Query Model}\label{Subsection: Query Algorithm}

In this section, we describe how to realize the algorithm in the query model. Here, we can directly use $\degreeq{}$ and $\randedgeq{}$ to construct the sampled set of vertices $\vertexsample$, and the sampled set of edges $\edgesample$.

\begin{algorithm}[ht!]
    \caption{\queryccdh}\label{Algorithm: Query CCDH Approx}
    \begin{algorithmic}[1]
        \REQUIRE $\vertexcount,\edgecount,\approxhindex  \leq \hindex,\approxdom,\approxran$
        \STATE $\vertexsamplesize \gets \frac{\constant\vertexcount}{\approxhindex\approxran^2}\log\vertexcount$
        \STATE $\edgesamplesize \gets \frac{\constant\edgecount}{\approxhindex\approxdom^2}\log\vertexcount$
        \STATE $\vertexsampleindex \gets \vertexsamplesize$ i.i.d. samples from $\uniform\fbrac{[\vertexcount]}$
        \STATE $\degreesample \gets \emptyset$
        \STATE $\edgesample \gets \emptyset$
        \STATE \textbackslash\textbackslash Making Queries
        \FOR{$i \in \vertexsampleindex$}
            \STATE $\degree{i} \gets \degreeq(i)$
            \STATE $\degreesample \gets \degreesample \cup \sbrac{\degree{i}}$
        \ENDFOR
        \FOR{$i \in \tbrac{\edgesamplesize}$}
            \STATE $\edgesample \gets \edgesample \cup \sbrac{\randedgeq{()}}$
        \ENDFOR
        \STATE \textbackslash\textbackslash Outputting the estimate
        \STATE $\approxdegree{\vertex} = \frac{\edgecount\size{\sbrac{\edge \in \edgesample|\vertex \in \edge}}}{\edgesamplesize}, \forall \vertex \in \vertexset$
        \IF{$\degree{} \leq \approxhindex$}
            \STATE $\approxccdh(\degree{}) \gets \frac{\vertexcount\size{\sbrac{\degree{i} \in \degreesample|\degree{i} \geq \degree{}}}}{\vertexsamplesize}$
        \ENDIF
        \IF{$\degree{} > \approxhindex$}
            \STATE $\approxccdh(\degree{}) \gets \size{\sbrac{\vertex|\approxdegree{\vertex}\geq \degree{}}}$
        \ENDIF
    \end{algorithmic}
\end{algorithm}

\begin{theorem}[{\queryccdh{}} Guarantee]\label{Theorem: Query CCDH Algorithm Guarantee}
    \queryccdh{} uses $\bigot{\nicefrac{\vertexcount}{\approxhindex\approxran^2}+\nicefrac{\edgecount}{\approxhindex\approxdom^2}}$ non-adaptive queries, and its output $\approxccdh$ is an $\biapprox$-BMA of the original $\ccdh$ with high probability.
\end{theorem}

\begin{proof}
    By design, \queryccdh{} uses $\bigot{\nicefrac{\vertexcount}{\hindex\approxran^2}}$ \degreeq{} queries, and uses $\bigot{\nicefrac{\edgecount}{\approxhindex\approxdom^2}}$ \randedgeq{} queries. Combined, the algorithm uses $\bigot{\nicefrac{\vertexcount}{\approxhindex\approxran^2}+\nicefrac{\edgecount}{\approxhindex\approxdom^2}}$ queries.
    The approximation guarantee again follows directly from Theorem~\ref{Theorem: General CCDH Algorithm Guarantee}.
\end{proof}

Next, we present an adaptive query algorithm \updqueryccdh{} that uses $\widetilde{O}\!\left(\nicefrac{\edgecount}{\hindex}\!\left(\nicefrac{1}{\approxdom^2}+\nicefrac{1}{\approxran^2}\right)\right)$ queries. Observe that when $\edgecount = \Omega(\vertexcount)$, the \magccdh{} algorithm uses
$\widetilde{O}\!\left(\nicefrac{\edgecount}{\hindex}\!\left(\nicefrac{1}{\approxdom^2}+\nicefrac{1}{\approxran^2}\right)\right)$ samples. Hence, we consider the case when $\edgecount = \smallo{\vertexcount}$. As before, the idea is to sample from the set of active vertices $\actverset$.

We exploit the fact that endpoints of uniformly sampled edges are always active to design an efficient \emph{Active Vertex Sampling} procedure \actver{}.
The procedure samples an edge uniformly at random, selects one of its endpoints uniformly, and outputs the selected vertex with probability $1/\degree{\vertex}$ (otherwise returning $\perp$). We now present the algorithm and state its performance guarantees.

\begin{algorithm}
    \caption{\actver{}}
    \begin{algorithmic}[1]\label{Algorithm: Active Vertex Sampling}
        \STATE Select an edge $\edge \in \edgeset$ uniformly at random. \label{line:randedge}
        \STATE Let $\vertex$ be a uniformly random endpoint of $\edge$. \label{line:randvertex}
        \STATE With probability $\frac{1}{\degree{\vertex}}$, return $\vertex$; otherwise, return $\perp$. \label{line:output}
    \end{algorithmic}
\end{algorithm}

\begin{restatable}[\actver{} Guarantee]{lemma}{actversamp}\label{Lemma: Active Vertex Sampling Probability}
    \actver{} reports either a vertex $\vertex \in \actverset$ with probability $\nicefrac{1}{2\edgecount}$ or $\perp$ is reported with probability $1-\nicefrac{\actvercount}{2\edgecount}$. In particular, if a vertex $\vertex \in \actverset$ is reported, then it is selected uniformly at random from the set $\actverset$.
\end{restatable}

\begin{proof}
From the description of \Cref{Algorithm: Active Vertex Sampling}, it is clear that the algorithm outputs either a vertex $\vertex \in \actverset$ or $\perp$. Fix a vertex $\vertex \in \actverset$. The probability that an edge incident to $\vertex$ is selected in \hyperlink{line:randedge}{Line 1} is $\degree{\vertex}/\edgecount$. 

Given that such an edge is chosen, the probability that $\vertex$ is selected in \hyperlink{line:randvertex}{Line 2} is $1/2$, and then it is kept in Line 3 with probability $1/\degree{\vertex}$. Hence, the probability that $\vertex$ is reported in Line 3 is $\degree{\vertex}/\edgecount \cdot 1/2 \cdot 1/\degree{\vertex} = 1/2\edgecount$. 

Since these events are mutually exclusive over all $\vertex \in \actverset$, the probability that the algorithm outputs a vertex in $\actverset$ is $ \actvercount/2\edgecount$. Therefore, the probability that the output is $\perp$ is $1 - \actvercount/2\edgecount$.

Since the algorithm reports some vertex in $\actverset$ with probability $\actvercount / 2\edgecount$ and reports a particular vertex $\vertex \in \actverset$ with probability $1 / 2\edgecount$, we can use conditional probability to claim that the probability that the output is $\vertex$ given that it is some vertex in $\actverset$ is $\frac{1 / 2\edgecount}{\actvercount / 2\edgecount} = \frac{1}{\actvercount}$.
\end{proof}

We first design a generic algorithm, \updccdh{}, analogous to \magccdh{}, with modifications limited to \hyperlink{line1}{Line 1} and \hyperlink{line6}{Line 6} using the \actver{} procedure. The next result states the estimation guarantees of \updccdh{}. The algorithm and the proofs are provided in \Cref{AppSubsec: Query Algo Implementation}.


\begin{restatable}[\updccdh{} Guarantee]{theorem}{updmagalgo}\label{Theorem: Updated CCDH Algorithm Guarantee}
    \updccdh{} uses $\bigo{(\nicefrac{\edgecount}{\approxhindex\approxran^2}+\nicefrac{\edgecount}{\approxhindex\approxdom^2})\log\vertexcount}$ degree and edge samples and its output $\approxccdh$ is an $\biapprox$-BMA of the original $\ccdh$ with high probability.
\end{restatable}

We now outline how to realise the \updccdh{} algorithm in the adaptive query model. We first obtain a uniformly random edge using one \randedgeq{}
query. One endpoint is chosen uniformly at random, and its degree
$\degree{}$ is obtained using a \degreeq{} query. The vertex is
retained with probability $1/\degree{}$. Thus, Active Vertex
Sampling can be simulated using $O(1)$ queries with two rounds of
adaptivity. The algorithm \updqueryccdh{} and the analysis is presented in \Cref{AppSubsec: Query Algo Implementation}.

\begin{restatable}[\updqueryccdh{} Guarantee]{theorem}{qryupdthm}\label{Theorem: Query Guarantee Updated}
\updqueryccdh{} uses $\bigot{\nicefrac{\edgecount}{\approxhindex\approxran^2}+\nicefrac{\edgecount}{\approxhindex\approxdom^2}}$ queries with two rounds of adaptivity, and its output $\approxccdh$ is an $\biapprox$-BMA of the ccdh. 
\end{restatable}

\begin{remark}[Discussion on approximate h-index $\approxhindex$]\label{Remark: Discussion on h-index}
    In this section, we assume the algorithms have access to an approximate h-index $\approxhindex \leq \hindex$. This is a standard practice in both graph streaming and query literature~\cite{Kane/ICALP/Subgraph/2012,McGregorVV/pods/2016/BetterTriangleCountingStreams,KallaugherP/soda/2017/HybridTriangleStream,DBLP:conf/stoc/EdenRS18,KallaugherKP/focs/2018/SketchingSubgraphCounting,KallaugherMPV/pods/2019/ALModelCycleCount}. Setting $\approxhindex = \bigomega{\hindex}$ recovers the bounds stated in \Cref{Tab: Results Table}. Furthermore, the algorithm due to \cite{AssadiNguyen/Approx/2022/OptimalH-IndexAndTriangle} can be implemented in $\bigo{\log\vertexcount}$ passes in the streaming model, and with $\bigo{\log\vertexcount}$ rounds of adaptivity in the query model.
\end{remark}


\section{Lower Bounds}\label{Section: Lower Bounds}

To complement our upper bounds, we now establish fundamental limits on the complexity of approximating the ccdh. 
These lower bounds show that sublinear estimation is impossible in general, and the dependence of our algorithms on key graph parameters is unavoidable. 
Our results apply to adaptive query and multi-pass streaming models. We provide more details in \Cref{App: Proofs Lower Bounds}. Our lower bounds use the communication complexity problem of set-disjointness defined below.


\begin{restatable}[Set-Disjointness($\disjointness{}$)]{definition}{disjdef}\label{Definition: Disjointedness Problem}
    Given two strings $\alicestring, \bobstring \in \sbrac{0,1}^\stringlength$, we define the $\disjointness{\stringlength}$ function that returns $\disjointness{\stringlength}$ $(\alicestring,\bobstring) = 1$ if $\sum_{i \in [\stringlength]} \alicestring_i \bobstring_i = 0$, and  $\disjointness{\stringlength}(\alicestring,\bobstring) = 0$ if $\sum_{i \in [\stringlength]} \alicestring_i \bobstring_i \ne 0$.
\end{restatable}

We specifically consider the promise version of disjointness, $\promisedisjointness{\stringlength}$ where we are promised $\sum_i \alicestring_i\bobstring_i \in \sbrac{0,1}$. In \Cref{Subsection: General Graph Lower Bound}, we prove a $\bigomega{\vertexcount}$ lower bound for general graphs. 
This bound is tight, as a trivial $\bigo{\vertexcount}$ algorithm exists by querying the degree of every vertex in the query model or explicitly maintaining vertex degrees in the streaming model. In \Cref{Subsection: h-index Lower Bound}, we establish a $\bigomega{\nicefrac{\edgecount}{\hindex}}$ lower bound, which is optimal up to polylogarithmic factors in terms of graph parameters.

\subsection{General Lower Bound}\label{Subsection: General Graph Lower Bound}

In this section, we establish the $\bigomega{\vertexcount}$ lower bounds for the general graph case across both models using the same graph construction. We start with a brief overview of the construction. Given an index $i \in \tbrac{\stringlength}$, we define a gadget on five vertices labeled $5i, 5i+1, \ldots, 5i+4$. If $\alicestring_i = 1$, we add an edge between vertices $5i$ and $5i+1$; otherwise, if $\alicestring_i = 0$, we add an edge between $5i+1$ and $5i+3$.
Similarly, if $\bobstring_i = 1$, we add an edge between vertices $5i$ and $5i+2$; otherwise, if $\bobstring_i = 0$, we add an edge between $5i+2$ and $5i+4$.
Figure~\ref{Figure: Gadget for Lower Bound General Graph} illustrates the resulting gadget.

Observe that if $\promisedisjointness{\stringlength}(\alicestring,\bobstring)=1$, then every vertex in the constructed graph has degree either $0$ or $1$. Otherwise, there is exactly one vertex of degree $2$. We exploit this gap in the degree distribution to reduce the $\promisedisjointness{}$ problem to distinguishing between the ccdh of these two graph instances. We now state our lower bounds for general graphs. 

\begin{figure}[ht!]
\centering

\subfloat[$x_i=0,y_i=0$\label{fig:1a}]{%
\begin{tikzpicture}[scale=1,framed]
\path[use as bounding box] (-0.6,-1.8) rectangle (5.6,1.8);
\tikzstyle{every node}=[align=center,inner sep=0.5pt,minimum size=2.5em]
\node[shape=circle,draw] (A) at (0,0) {$5i$};
\node[shape=circle,draw] (B) at (2.25,1) {$5i$\\$+1$};
\node[shape=circle,draw] (C) at (2.25,-1) {$5i$\\$+2$};
\node[shape=circle,draw] (D) at (5,1) {$5i$\\$+3$};
\node[shape=circle,draw] (E) at (5,-1) {$5i$\\$+4$};
\path [-] (B) edge node[midway,above] {$x_i:0$} (D);
\path [-] (C) edge node[midway,below] {$y_i:0$} (E);
\end{tikzpicture}%
}\hfill
\subfloat[$x_i=0,y_i=1$\label{fig:1b}]{%
\begin{tikzpicture}[scale=1,framed]
\path[use as bounding box] (-0.6,-1.8) rectangle (5.6,1.8);
\tikzstyle{every node}=[align=center,inner sep=0.5pt,minimum size=2.5em]
\node[shape=circle,draw] (A) at (0,0) {$5i$};
\node[shape=circle,draw] (B) at (2.25,1) {$5i$\\$+1$};
\node[shape=circle,draw] (C) at (2.25,-1) {$5i$\\$+2$};
\node[shape=circle,draw] (D) at (5,1) {$5i$\\$+3$};
\node[shape=circle,draw] (E) at (5,-1) {$5i$\\$+4$};
\path [-] (A) edge node[midway,below,rotate=-23] {$y_i:1$} (C);
\path [-] (B) edge node[midway,above] {$x_i:0$} (D);
\end{tikzpicture}%
}\hfill

\vspace{15pt}

\subfloat[$x_i=1,y_i=0$\label{fig:1c}]{%
\begin{tikzpicture}[scale=1,framed]
\path[use as bounding box] (-0.6,-1.8) rectangle (5.6,1.8);
\tikzstyle{every node}=[align=center,inner sep=0.5pt,minimum size=2.5em]
\node[shape=circle,draw] (A) at (0,0) {$5i$};
\node[shape=circle,draw] (B) at (2.25,1) {$5i$\\$+1$};
\node[shape=circle,draw] (C) at (2.25,-1) {$5i$\\$+2$};
\node[shape=circle,draw] (D) at (5,1) {$5i$\\$+3$};
\node[shape=circle,draw] (E) at (5,-1) {$5i$\\$+4$};
\path [-] (A) edge node[midway,above,rotate=23] {$x_i:1$} (B);
\path [-] (C) edge node[midway,below] {$y_i:0$} (E);
\end{tikzpicture}%
}\hfill
\subfloat[$x_i=1,y_i=1$\label{fig:1d}]{%
\begin{tikzpicture}[scale=1,framed]
\path[use as bounding box] (-0.6,-1.8) rectangle (5.6,1.8);
\tikzstyle{every node}=[align=center,inner sep=0.5pt,minimum size=2.5em]
\node[shape=circle,draw=red] (A) at (0,0) {$5i$};
\node[shape=circle,draw] (B) at (2.25,1) {$5i$\\$+1$};
\node[shape=circle,draw] (C) at (2.25,-1) {$5i$\\$+2$};
\node[shape=circle,draw] (D) at (5,1) {$5i$\\$+3$};
\node[shape=circle,draw] (E) at (5,-1) {$5i$\\$+4$};
\path [-] (A) edge node[midway,above,rotate=23] {$x_i:1$} (B);
\path [-] (A) edge node[midway,below,rotate=-23] {$y_i:1$} (C);
\end{tikzpicture}%
}

\caption{Gadget for the general lower bound (\Cref{Theorem: Query General LB,Theorem: Streaming General LB})}
\label{Figure: Gadget for Lower Bound General Graph}
\end{figure}


\begin{restatable}[General Lower Bound - Streaming]{theorem}{StrmGenLB}\label{Theorem: Streaming General LB}

  Let $\vertexcount \in \fN$. Any algorithm that takes as input a graph $\graph$ with $\vertexcount$ vertices in the multi-pass streaming model and computes a $\biapprox$-BMA of ccdh of $\graph$ with probability at least $2/3$, for any $\approxdom < \frac{1}{2}$ and $\approxran < 1$, must use $\bigomega{\vertexcount}$ space.
\end{restatable}

\begin{proof}
    We use the gadgets described above to reduce from the $\promisedisjointness{\stringlength}$ problem. Given an input $\alicestring$, Alice generates the edge stream of $\stream^A = \fbrac{\streamelement_1,\ldots,\streamelement_\stringlength}$ where $\streamelement_i = \fbrac{5i,5i+1}$ if $\alicestring_i = 1$, and $\fbrac{5i+1,5i+3}$ otherwise. Correspondingly, Bob generates the edge stream $\stream^B = \fbrac{\edge_{\stringlength+1},\ldots,\edge_{2\stringlength}}$ where $\streamelement_{\stringlength+i} = \fbrac{5i,5i+2}$ if $\bobstring_i = 1$, and $\fbrac{5i+2,5i+4}$ otherwise. 
    
    Now, observe that the graph has $\vertexcount = 5\stringlength$ vertices, and $\edgecount = 2\stringlength$ edges. Furthermore, if $\promisedisjointness{\stringlength}\fbrac{\alicestring,\bobstring} = 1$, then we have $\ccdh(i) = \nicefrac{4\vertexcount}{5}$ if $i \leq 1$, and $0$ otherwise. Alternatively, if $\promisedisjointness{\stringlength}\fbrac{\alicestring,\bobstring} = 0$, then we have $\ccdh(i) = \nicefrac{4\vertexcount}{5}$ if $i \leq 1$, $\ccdh(i) = 1$ if $i = 2$, and $0$ otherwise.

Let $\algo$ be a streaming algorithm that makes $p$ passes and uses $\streamspace$ space to compute a $\biapprox$-BMA of the CCDH with success probability at least $2/3$. As the algorithm $\algo$ computes an $\biapprox$-BMA $\approxccdh$ for $\approxdom < \frac{1}{2}$, we would have:
    \begin{align*}
        \approxccdh(2) = \begin{dcases}
            >0 & if \text{ }\promisedisjointness{\stringlength}\fbrac{\alicestring,\bobstring} = 0\\
            0 & if \text{ }\promisedisjointness{\stringlength}\fbrac{\alicestring,\bobstring} = 1
        \end{dcases}
    \end{align*}

    Then, we can solve $\promisedisjointness{\stringlength}$ with probability $2/3$ and using $\bigo{ p \cdot \streamspace}$ bits.~\footnote{This is because Alice and Bob can simulate the streaming algorithm by communicating its memory state. Hence, a $p$-pass streaming algorithm can be simulated using at most $2p$ communications of the memory state.} From Lemma~\ref{Lemma: Hardness of Promise Disjointness}, we obtain $\streamspace = \bigomega{\vertexcount/p}$, which is $\Omega(\vertexcount)$ as long as $p$ is a constant.
\end{proof}

\begin{restatable}[General Lower Bound - Query]{theorem}{QueryGenLB}\label{Theorem: Query General LB}

     Let $\vertexcount \in \fN$. Any algorithm using $\degreeq{}$, $\neighbourq{}$, $\edgeexistsq{}$, and $\randedgeq{}$ query access to a graph $\graph$ with $\vertexcount$ vertices that computes a $\biapprox$-BMA of the ccdh of $\graph$ with probability at least $2/3$, for any $\approxdom < \frac{1}{2}$ and $\approxran < 1$, must use $\bigomega{\vertexcount}$ queries.
\end{restatable}

\begin{proof}
    The embedding $\embedding$ that we use here is as described above in Figure~\ref{Figure: Gadget for Lower Bound General Graph}. For any $\alicestring,\bobstring \in \sbrac{0,1}^\stringlength$, $\embedding(\alicestring,\bobstring)$ creates a graph $\ccgraph = \fbrac{\ccvertexset,\ccedgeset}$ with $5\stringlength$ vertices labeled $\cup_{i\in\tbrac{\stringlength}} \{5i,5i+1,5i+2,5i+3,5i+4\} = \ccvertexset$. For the edges in $\ccedgeset$, for each $i \in \tbrac{\stringlength}$, $\fbrac{5i,5i+1} \in \ccedgeset$ if $\alicestring_i = 1$, $\fbrac{5i+1,5i+3} \in \ccedgeset$ if $\alicestring_i = 0$. Similarly, we also add for each $i \in \tbrac{\stringlength}$, $\fbrac{5i,5i+2} \in \ccedgeset$ if $\bobstring_i = 1$, $\fbrac{5i+2,5i+4} \in \ccedgeset$ if $\bobstring_i = 0$. Now, as stated earlier, we have for this construction:
    \begin{align*}
        \approxccdh_{\ccgraph}(2) = \begin{dcases}
            >0 & \text{if  }\promisedisjointness{\stringlength}\fbrac{\alicestring,\bobstring} = 0\\
            0 & \text{if  }\promisedisjointness{\stringlength}\fbrac{\alicestring,\bobstring} = 1
        \end{dcases}
    \end{align*}
    An algorithm that constructs an $\biapprox$-BMA of ccdh for any $\approxdom < \frac{1}{2}$ for all graphs can distinguish between these two cases. We discuss how to implement the queries in the communication model:
    \begin{itemize}
        \item \textbf{\degreeq{}: } Given a vertex $j = 5i+k$, $0 \leq k \leq 4$, checking $\alicestring_i$ and $\bobstring_i$ suffices to know the degree. Therefore, we need at most $2$ bits of communication to implement the \degreeq{} query.
        \item \textbf{\randedgeq{}: } Each edge in the graph has a one-to-one correspondence with each element in the two strings. Choose one index out of the $[2\stringlength]$ indices at random using the public randomness and communicate that bit to obtain the corresponding edge. Hence, we need $1$ bit of communication to implement the \randedgeq{} given access to shared random string to Alice and Bob. 
        \item \textbf{\edgeexistsq{}: } Each vertex-pair has at most one possible way of having an edge, defined by corresponding $\alicestring_i$ or $\bobstring_i$. Communicating this bit suffices to decide whether the given vertex-pair contain an edge or not. Thus, communicating $1$ bit suffices to implement \edgeexistsq{} query. 
        \item \textbf{\neighbourq{}: } Each vertex has at most two neighbours, depending on corresponding $\alicestring_i$ and $\bobstring_i$. Communicating the bits suffices to obtain all possible neighbours of the vertex. Hence, communicating $2$ bits suffice to implement the \neighbourq{} query. 
    \end{itemize}
    As described above, all the queries can be implemented using $\bigo{1}$ bits of communication. Thus, any algorithm computing an $\biapprox$-BMA of ccdh for any $\approxdom < \frac{1}{2}$ requires to make $\bigomega{\stringlength}$ queries. Since $\vertexcount = 5\stringlength$, we have the stated lower bound.
\end{proof}

\subsection{h-index Based Lower Bound}\label{Subsection: h-index Lower Bound}

Observe that all graphs constructed in \Cref{Subsection: General Graph Lower Bound} have h-index equal to $1$. Consequently, these bounds do not preclude the possibility of more efficient algorithms for graphs with larger h-index. In this section, we therefore establish h-index–sensitive lower bounds. We show that for any $\vertexcount, \edgecount, \hindex \in \Nat$ satisfying $\vertexcount \ge \nicefrac{5\edgecount}{\hindex}$, any streaming or query algorithm computing a $\biapprox$-BMA of the ccdh on graphs with $\vertexcount$ vertices, $\edgecount$ edges, and h-index $\hindex$ must use $\bigomega{\nicefrac{\edgecount}{\hindex}}$ space or queries, respectively.

We begin with a brief overview of the graph construction used throughout this section.
We consider the $\promisedisjointness{}$ problem on strings of length $\stringlength$.
Unlike the previous lower bound, which constructs a separate gadget for each
$i \in \tbrac{\stringlength}$, here we build a single graph
(see Figure~\ref{Figure: Gadget for Lower Bound h index})
using five main vertex sets
$\vertexsetA, \vertexsetB, \vertexsetC, \vertexsetD, \vertexsetE$,
each of size $\stringlength$, and an additional set of isolated vertices
$\isovertexset$ of size $\vertexcount - 5\stringlength$, which plays no role in the construction.


Fix an arbitrary enumeration of the vertices in
$\vertexsetA, \vertexsetB,\vertexsetC,$ $\vertexsetD, \vertexsetE$,
denoted by $\vertexA_i, \vertexB_i, \vertexC_i, \vertexD_i,$ and $\vertexE_i$
for $i \in \tbrac{\stringlength}$.
If $\alicestring_i = 1$, then $\vertexA_i$ is connected to
$\vertexB_j$ for $j \in [i,\, i + \nicefrac{3\hindex}{4} - 1]$,
and $\vertexD_i$ is connected to $\vertexB_j$ for
$j \in [i + \nicefrac{3\hindex}{4},\, i + \hindex - 1]$.
If instead $\alicestring_i = 0$, then $\vertexA_i$ is connected to
$\vertexB_j$ for $j \in [i,\, i + \nicefrac{\hindex}{4} - 1]$,
and $\vertexD_i$ is connected to $\vertexB_j$ for
$j \in [i + \nicefrac{\hindex}{4},\, i + \hindex - 1]$.
All index arithmetic is modulo $\stringlength$. Edges corresponding to $\bobstring$ are added analogously by replacing
$\alicestring$, $\vertexsetB$, and $\vertexsetD$ with
$\bobstring$, $\vertexsetC$, and $\vertexsetE$, respectively.

\begin{wrapfigure}{r}{0.5\textwidth}
\centering

\begin{tikzpicture}[scale=0.75]
\node[shape=circle,draw,minimum height=3.5em] (A) at (0,0) {$A, \stringlength$};
\node[shape=circle,draw,minimum height=3.5em] (B) at (4,2) {$B, \stringlength$};
\node[shape=circle,draw,minimum height=3.5em] (C) at (4,-2) {$C, \stringlength$};
\node[shape=circle,draw,minimum height=3.5em] (D) at (8,2) {$D, \stringlength$};
\node[shape=circle,draw,minimum height=3.5em] (E) at (8,-2) {$E, \stringlength$};

\path [-] (A) edge node[midway,above, rotate=23] {$x_i : 1$, $\frac{3\hindex}{4}$} node[midway,below, rotate=23] {$x_i : 0$, $\frac{\hindex}{4}$} (B);
\path [-] (A) edge node[midway,above, rotate=-23] {$y_i : 1$, $\frac{3\hindex}{4}$} node[midway,below, rotate=-23] {$y_i : 0$, $\frac{\hindex}{4}$} (C);
\path [-] (B) edge node[midway,below] {$x_i : 0$, $\frac{3\hindex}{4}$} node[midway,above] {$x_i : 1$, $\frac{\hindex}{4}$} (D);
\path [-] (C) edge node[midway,below] {$y_i : 0$, $\frac{3\hindex}{4}$} node[midway,above] {$y_i : 1$, $\frac{\hindex}{4}$} (E);
\end{tikzpicture}

\caption{Gadget for $h$-index sensitive Lower Bound (\Cref{Theorem: Query h-index Lower Bound,Theorem: Streaming h-index Lower Bound}).}

\label{Figure: Gadget for Lower Bound h index}
\end{wrapfigure}

Each vertex $\vertexB_i$ (respectively, $\vertexC_i$) is incident to one edge for each
$\alicestring_j$ (respectively, $\bobstring_j$) with
$j \in \{i-\hindex+1,\ldots,i\}$.
Since $\stringlength \ge \hindex$, every vertex in $\vertexsetB$ and $\vertexsetC$
has degree exactly $\hindex$.
All vertices in $\vertexsetD$ and $\vertexsetE$ have degree at most $\nicefrac{3\hindex}{4}$. For $\vertexsetA$, if $\promisedisjointness{\stringlength}(\alicestring,\bobstring)=1$,
then all vertices have degree at most $\hindex$.
Otherwise, there exists a unique index $j$ with $\alicestring_j=\bobstring_j=1$,
and the corresponding vertex $\vertexA_j$ has degree $\nicefrac{6\hindex}{4}$.
Thus, in all cases, the graph contains at least $2\hindex$ vertices of degree exactly
$\hindex$ and at most one vertex with degree greater than $\hindex$.
Consequently, every graph constructed has h-index exactly $\hindex$. Observe that if $\promisedisjointness{\stringlength}(\alicestring,\bobstring)=1$, then every vertex in the constructed graph has degree at most $\hindex$. Otherwise, there is exactly one vertex of degree $\nicefrac{6\hindex}{4}$.

We now state our lower bounds for h-index-parametrized graphs. The proofs are in \Cref{AppSubSec: hindex LB}.

\begin{restatable}[h-index Sensitive Lower Bound - Streaming]{theorem}{StrmhindLB}\label{Theorem: Streaming h-index Lower Bound}

   Let  $\vertexcount,\edgecount,\hindex \in \Nat$ such that $\vertexcount \geq \frac{5\edgecount}{\hindex}$. Any algorithm that takes as input a graph $\graph$ with $\vertexcount$ vertices, $\edgecount$ edges, and h-index $\hindex$ in the multi-pass streaming model and computes a $\biapprox$-BMA of the ccdh of $\graph$ with probability at least $2/3$, for any $\approxdom < \frac{1}{5}$ and $\approxran < 1$, must use $\bigomega{\frac{\edgecount}{\hindex}}$ space.
\end{restatable}

\begin{restatable}[h-index Sensitive Lower Bound - Query]{theorem}{QueryhindLB}\label{Theorem: Query h-index Lower Bound}

Let $\vertexcount,\edgecount,\hindex \in \Nat$ be such that $\vertexcount \geq \frac{5\edgecount}{\hindex}$. Any algorithm using $\degreeq{}$, $\neighbourq{}$, $\edgeexistsq{}$, and $\randedgeq{}$ query access to a graph $\graph$ with $\vertexcount$ vertices, $\edgecount$ edges, and h-index $\hindex$ that computes a $\biapprox$-BMA of the ccdh of $\graph$ with probability at least $2/3$, for any $\approxdom < \frac{1}{5}$ and $\approxran < 1$, must use $\bigomega{\frac{\edgecount}{\hindex}}$ queries.
\end{restatable}

\begin{remark}[Optimality of Eden et al.~\cite{Eden/WWW/2017/CCDHinQueryModel}]
Eden et al.~\cite{Eden/WWW/2017/CCDHinQueryModel} proposed an algorithm in the adaptive query model that computes a
$(\approxerror,\approxerror)$-BMA of the ccdh using
$\bigot{\nicefrac{\vertexcount}{\approxerror^2\hindex}+\nicefrac{\edgecount}{\approxerror^2z^{2}}}$
\degreeq{} and \neighbourq{} queries.
Here, $z$ is a graph parameter, termed the \emph{$z$-index}, introduced in their work and defined as $z = \min_{\{\degree{} \mid \ccdh(\degree{}) > 0\}} \sqrt{\degree{} \cdot \ccdh(\degree{})}$,
which satisfies the bound $z \ge \sqrt{\hindex}$.
Consequently, our $\bigomega{\edgecount/\hindex}$ lower bound implies that this algorithm is optimal up to polylogarithmic factors. 
Note that  our algorithm uses a different set of queries (\degreeq{}+\randedgeq{}) than Eden et al.~\cite{Eden/WWW/2017/CCDHinQueryModel}~(\degreeq{}+\neighbourq{}). Neither of these models subsumes the other. In particular simulating \randedgeq{} using \degreeq{} and \neighbourq{} queries require $\bigomega{\sqrt{\nicefrac{\vertexcount}{\degreeq{}}}}$ queries, and simulating \neighbourq{} queries using \degreeq{} and \randedgeq{} queries require $\Omega(\edgecount)$ queries. 
\end{remark}

\begin{remark}[Lower Bounds Extension to Multisets]\label{Remark: Lower Bound Generalization}
Recall the discussion in \Cref{Remark: Non-Graph Extension}. By using the construction above and noting that an analogous construction applies to multisets, the lower bounds naturally extend to the more general problem of learning arbitrary multisets, as described.
\end{remark}

\section{Non-adaptive Query Lower Bound}\label{SubSec: Non_Adaptive LB}

In this section, we present a $\bigomega{\min\sbrac{\nicefrac{\edgecount}{\log\hindex},\nicefrac{\vertexcount}{\hindex}}}$ lower bound for non-adaptive query algorithms for the case of sparse graphs. 

\begin{restatable}[h-index Sensitive Lower Bound - Non-adaptive Query]{theorem}{QueryNonAdaptiveLB}\label{Theorem: Query Non-adaptive Lower Bound}
Let $\vertexcount,\edgecount,\hindex \in \Nat$ be such that $\hindex$ is sufficiently large, $\edgecount \geq \hindex^2$, and $\vertexcount \geq 2\edgecount$. Any non-adaptive algorithm using $\degreeq{}$, $\neighbourq{}$, $\edgeexistsq{}$, and $\randedgeq{}$ query access to a graph $\graph$ with $\vertexcount$ vertices, $\edgecount$ edges, and h-index $\hindex$ that computes a $\biapprox$-BMA of the ccdh of $\graph$ with probability at least $2/3$, for any $\approxdom < \frac{1}{2}$ and $\approxran \leq c$, where $c>0$ is a sufficiently small absolute constant, must use either $\bigomega{\frac{\vertexcount}{\hindex}}$ $\degreeq{}$, $\neighbourq{}$, and $\edgeexistsq{}$ queries, or $\bigomega{\frac{\edgecount}{\log \hindex}}$ $\randedgeq{}$ queries. In particular, it must use $\bigomega{\min\left\{\frac{\vertexcount}{\hindex}, \frac{\edgecount}{\log \hindex}\right\}}$ queries.
\end{restatable}

Our proof uses the following problem: consider an urn of $k$ balls, each having a color. The color multiplicities are
arbitrary. Let $D$ denote the number of distinct colors in the urn. Our objective
is to estimate $D$ by sampling balls without replacement, that is, each sample
reveals the color of a ball chosen uniformly at random among the balls not drawn
before.

\begin{proposition}[Wu--Yang {\cite[Theorem~4]{WY18}}]\label{prop:wy}
Let $k\in\mathbb N$ be sufficiently large, and let $\varepsilon>0$ be at most a sufficiently small absolute constant.
Let $\mathcal A$ be a (possibly randomized) algorithm that draws $N$ samples
without replacement from an urn of $k$ colored balls and, for every such urn,
outputs $\widehat D$ with $|\widehat D-D|\le\varepsilon k$ with probability at
least $3/5$, where $D$ is the number of distinct colors. Then $N=\Omega(k/\log k)$.
\end{proposition}

\begin{remark}
In \cite{WY18}, $n^*(k,\Delta)$ denotes the least number of samples drawn with
replacement for which some estimator $\widehat C$ satisfies
$\Pr[|\widehat C-C|\ge\Delta]\le 0.1$ on every urn of $k$ balls.
\cite[Theorem~4]{WY18} states that for a sufficiently small constant $c$ and all
$1\le\Delta\le ck$, $n^*(k,\Delta)=\Omega\big(\frac{k}{\log k}\log\frac{k}{\Delta}\big)$,
and that the same bound holds without replacement.
Proposition~\ref{prop:wy} is the case $\Delta=\lfloor 2\varepsilon k\rfloor+1$, since
rounding $\widehat D$ to the nearest integer at most doubles its error. Here
$1\le\Delta\le ck$ for small $\varepsilon$ and large $k$, so $\log(k/\Delta)=\Omega(1)$.
The success probability $0.9$ in \cite{WY18} can be replaced by $3/5$: the lower bound
is proved by exhibiting two priors whose observations are close in total variation
distance, and this distance can be made an arbitrarily small constant by adjusting the
constant in $\Omega(\cdot)$.
\end{remark}

\subparagraph*{Overview of the proof of \Cref{Theorem: Query Non-adaptive Lower Bound}.}
Our hard instances consist of three vertex-disjoint parts, padded with isolated
vertices so that the graph has exactly $\vertexcount$ vertices. The first part, the \emph{core}, encodes an
urn of $k=\hindex$ colored balls with $D$ distinct colors: every color has one
vertex, and every ball is represented by three leaves, each adjacent only to the
vertex of the ball's color. Thus a color held by $f\geq 1$ balls yields a vertex of
degree $3f\geq 3$, and the vertices of unused colors are isolated. The second part
is a clique on $\hindex+1$ vertices, which makes the h-index of the graph exactly
$\hindex$. The third part is a perfect matching, which brings the number of edges
to exactly $\edgecount$. Every other vertex has degree at most $1$, so the vertices
of degree at least $2$ are exactly the $\hindex+1$ clique vertices and the $D$
color vertices of colors that occur; that is, $C(2)=D+\hindex+1$. Hence any
$\biapprox$-BMA $\widehat{C}$ of the ccdh yields the estimate
$\widehat{D}=\widehat{C}(2)-\hindex-1$ with
$|\widehat{D}-D|\leq\approxran(D+\hindex+1)\leq 3\approxran k$, an additive
$\varepsilon k$-estimate of $D$ with $\varepsilon=3\approxran$. The vertices are
labelled by a uniformly random permutation, so the core occupies a uniformly random
set of $4\hindex$ labels. A non-adaptive algorithm fixes the labels of its
$\degreeq{}$, $\neighbourq{}$ and $\edgeexistsq{}$ queries in advance, so $q_L$ such
queries touch the core with probability $\bigo{q_L\hindex/\vertexcount}$. Each
$\randedgeq{}$ query lands in the core with probability $3\hindex/\edgecount$, and
when it does it reveals the color of a single ball, while all other answers are
independent of the urn. Therefore, if $q_L=o(\vertexcount/\hindex)$, the algorithm
can be simulated by an estimator for $D$ that draws only
$\bigo{q_R\hindex/\edgecount}$ balls without replacement, and
Proposition~\ref{prop:wy} gives $q_R\hindex/\edgecount=\bigomega{\hindex/\log\hindex}$,
that is, $q_R=\bigomega{\edgecount/\log\hindex}$. 


\subparagraph{The hard instances}

.Let $\hindex\ge7$, $\edgecount\ge\hindex^2$ and $\vertexcount\ge2\edgecount$.
Given an urn $U$ of $\hindex$ colored balls with $D$ distinct colors, we construct a graph $\mathbf{G}_U$
on $\vertexcount$ vertices consisting of four vertex-disjoint parts.
\begin{itemize}
\item \textbf{Core.} The core has $\hindex$ \emph{color vertices} and $3\hindex$ \emph{leaves}.
Each color that occurs in the urn is assigned its own color vertex. Each ball is
represented by three leaves, each adjacent only to the color vertex of the ball's
color. Color vertices of colors that do not occur remain isolated.
\item \textbf{Clique.} A clique on $\hindex+1$ vertices.
\item \textbf{Matching.} A perfect matching with
$\edgecount-\frac{\hindex(\hindex+1)}{2}-3\hindex$ edges.
\item \textbf{Isolated vertices.} $\vertexcount-2\edgecount+\hindex^2+2\hindex-1$
isolated vertices.
\end{itemize}
Finally, the vertices are labelled by $1,\dots,\vertexcount$ according to a uniformly
random permutation of all $\vertexcount$ vertices, and adjacency lists are sorted by
label. Only the core depends on the urn; we call the clique, the matching and the
isolated vertices the \emph{fixed part}. 

\begin{observation}\label{lem:construction}
For every urn $U$ and every labelling, the graph  $\mathbf{G}_U$ satisfies the following.
\begin{enumerate}
\item It has $\vertexcount$ vertices, $\edgecount$ edges, and h-index $\hindex$.
\item The vertices of degree at least $2$ are exactly the $\hindex+1$ clique vertices
and the $D$ color vertices of colors that occur. Moreover, no vertex has degree $2$,
so $C(t)=D+\hindex+1$ for every real $t\in(1,3]$. Hence, if $\approxdom<\frac12$, every
$\biapprox$-BMA $\widehat C$ of the ccdh satisfies
$|\widehat C(2)-(D+\hindex+1)|\le3\approxran\hindex$.
\item The core has $4\hindex$ vertices and $3\hindex$ edges.
\end{enumerate}
\end{observation}

\subparagraph*{The augmented oracle for $\mathbf{G}_U$.}
The oracle knows the fixed part. It also knows the shape of the core. The core has
$\hindex$ color vertices, one available for each color. Each ball has three leaves,
adjacent to the color vertex of its color. The oracle has no direct access to the
urn $U$. It can only draw balls from $U$ without replacement.

First, the oracle chooses a uniformly random set of $4\hindex$ labels for the core.
It splits this set uniformly at random into $\hindex$ color-vertex labels and $3\hindex$
leaf labels. It labels the fixed part with the remaining labels, uniformly at random.

The core is then built gradually. Each time the oracle draws a ball, it notes the
ball's color. If the color is new, the oracle assigns it a uniformly random unused
color vertex. It then attaches three new leaves to the color vertex of that color. Their
labels are chosen uniformly at random among the unused leaf labels. Thus, after
$i$ draws, exactly $3i$ core edges are attached. These are the core edges of the
drawn balls.

\emph{Local queries.} Consider a $\degreeq{}$, $\neighbourq{}$ or $\edgeexistsq{}$ query.
If all its labels belong to the fixed part, the oracle answers it from the fixed part.
Otherwise, the oracle reports \emph{failure}.

\emph{Random edge queries.} Let $i$ be the number of balls drawn so far.
\begin{itemize}
\item With probability $1-3\hindex/\edgecount$, the oracle reports a uniformly random
edge of the fixed part.
\item With probability $\frac{3\hindex}{\edgecount}\cdot\frac{i}{\hindex}$, it reports a
uniformly random one of the $3i$ attached core edges.
\item Otherwise, it draws a new ball from $U$ without replacement. It attaches the
three leaves of this ball as described above. It then reports a uniformly
random one of these three new edges.
\end{itemize}
Given a parameter $N$, the oracle also reports \emph{failure} if it needs more than $N$ draws.

\begin{lemma}[Reduction]\label{lem:reduction}
Let $U$ be an urn of $\hindex$ balls from which we can sample balls without replacement, and let $\approxdom<\frac12$. Suppose a non-adaptive algorithm makes
$q_L =o(\frac{\vertexcount}{\hindex})$ local queries and $q_R=o(\frac{\edgecount}{\log \hindex})$ random edge queries to the augmented oracle for $\mathbf{G}_U$, and
outputs a $\biapprox$-BMA of the ccdh of $\mathbf{G}_U$ with probability at least $2/3$. Then there is
an algorithm that draws $o(\hindex/\log \hindex)$ balls without replacement from
$U$ and estimates $D$ within
$3\approxran\hindex$ with probability at least $3/5$.
\end{lemma}

\begin{proof}
Let $N\coloneqq\lceil C\,q_R\hindex/\edgecount\rceil$, where $C\ge300$ is a constant.
Then $N=o(\hindex/\log\hindex)$. Our algorithm for estimating $D$ proceeds in the
following steps.
\begin{enumerate}
\item Run the given algorithm against the augmented oracle for $\mathbf{G}_U$ with
parameter $N$. By construction, at most $N$ balls are drawn from $U$.
\item Assume that the oracle does not report failure. Let $\widehat C(2)$ be the value at
$d=2$ of the ccdh estimate output by the given algorithm.
\item Report $\widehat D = \widehat C(2)-\hindex-1$ as the estimate of $D$.
\end{enumerate}

\emph{Simulation.} As long as the oracle does not fail, its answers have the same
distribution as on $\mathbf{G}_U$. The labels are distributed as under a uniformly
random permutation. No edge joins the core to the fixed part. After $i$ draws, a
uniformly random core edge is one of the $3i$ attached edges with probability
$i/\hindex$. Otherwise, it belongs to a uniformly random undrawn ball. This is exactly
what a draw without replacement provides.

\emph{Approximation factor.} Let $\widehat C$ be a $\biapprox$-BMA with $\approxdom<\frac12$.
By Observation~\ref{lem:construction}(2),
$|\widehat D-D|=|\widehat C(2)-(D+\hindex+1)|\le\approxran(D+\hindex+1)\le3\approxran\hindex$.
So $\widehat D$ is an additive $\varepsilon k$-estimate with $k=\hindex$ and
$\varepsilon=3\approxran$. This $\varepsilon$ is small enough for
Proposition~\ref{prop:wy}.

\emph{Failure probability.} The labels of the local queries are fixed in advance. So
some local query involves a core label with probability at most
$2q_L\cdot4\hindex/\vertexcount=o(1)$.

Let $X$ denote the number of random edge queries that land in the core part. Each
random edge query lands in the core independently with probability
$3\hindex/\edgecount$. So $X$ is binomially distributed with $q_R$ trials, and
$\E[X]=3q_R\hindex/\edgecount$. By Markov's inequality,
\[
\Pr[X>N]\ \le\ \frac{\E[X]}{N+1}\ \le\ \frac{3q_R\hindex/\edgecount}{C\,q_R\hindex/\edgecount}\ =\ \frac{3}{C}\ \le\ 0.01 .
\]
The oracle draws at most one ball (without replacement) for each random edge query to the core part. So
the number of balls drawn without replacement is at most $X$. Hence the oracle draws
more than $N$ balls with probability at most $0.01$.
\end{proof}

\begin{proof}[Proof of Theorem~\ref{Theorem: Query Non-adaptive Lower Bound}]
Suppose, for contradiction, that the algorithm makes $q_L=o(\vertexcount/\hindex)$ local
queries and $q_R=o(\edgecount/\log\hindex)$ random edge queries. Let $U$ be any urn of
$\hindex$ balls. By Observation~\ref{lem:construction}(1), the graph $\mathbf{G}_U$ has
$\vertexcount$ vertices, $\edgecount$ edges and h-index $\hindex$. So the algorithm outputs
a $\biapprox$-BMA of the ccdh of $\mathbf{G}_U$ with probability at least $2/3$.

By Lemma~\ref{lem:reduction}, there is an algorithm that draws $o(\hindex/\log\hindex)$
balls without replacement from $U$. With probability at least $3/5$, it estimates $D$
within $3\approxran\hindex=\varepsilon k$, where $k=\hindex$ and $\varepsilon=3\approxran$.
Since $\approxran\le c$ for a sufficiently small constant $c$, the constant $\varepsilon$
is small enough for Proposition~\ref{prop:wy}. So Proposition~\ref{prop:wy} requires
$\Omega(\hindex/\log\hindex)$ draws, a contradiction.
Therefore $q_L=\bigomega{\vertexcount/\hindex}$ or $q_R=\bigomega{\edgecount/\log\hindex}$.
In particular, the algorithm must make
$\bigomega{\min\{\vertexcount/\hindex,\ \edgecount/\log\hindex\}}$ queries.
\end{proof}

\section{Conclusions}\label{Sec: Limitations}

We establish tight bounds for estimating the degree distribution in sublinear models, together with a model-agnostic framework based on degree and edge samples. Our algorithms assume access to exact degree and uniform edge samples, as is standard in these sublinear models; however, such access may be difficult to realize in some practical settings. Extending the framework to weaker or biased sampling access, as well as to other computational models, is an important direction for future work.


\clearpage

\bibliography{refs}




\clearpage

\appendix

\section{Proofs for\texorpdfstring{~\Cref{SubSec: MAgAlgo}}{4.1}}\label{App: Proofs for MAgAlgo}

In this section, we present the proofs for~\Cref{SubSec: MAgAlgo}. 

\subsection{Proof of \texorpdfstring{~\Cref{Lemma: General Algorithm High CCDH degree Guarantee}}{3}}

We start with~\Cref{Lemma: General Algorithm High CCDH degree Guarantee}.

\headest*

\begin{proof}
    For any degree $\degree{}$, for each vertex in $\vertexsample$, we assign a random variable $X_i, i \in [\vertexsamplesize]$ that takes value $1$ if the degree of the $i$-th vertex is $\geq \degree{}$, and $0$ otherwise. Then, we have for all degrees $\degree{}$ such that $\ccdh(\degree{}) \geq \ccdhthreshold$, $$\Prob\tbrac{X_i = 1} = \frac{\ccdh(\degree{})}{\vertexcount} \geq \frac{\ccdhthreshold}{\vertexcount}$$
    We denote $X = \nicefrac{1}{\vertexsamplesize} \sum_{i \in \vertexsamplesize} X_i$. Note that by line~\ref{Line: General Algorithm CCDH Based Approx} of Algorithm~\ref{Algorithm: Generalized CCDH Approx}, we have for all degrees $\degree{}$ such that $\ccdh(\degree{}) \geq \ccdhthreshold$, $\approxccdh(\degree{}) = \vertexcount X$. Then, by linearity of expectation and the fact that $X_i$ is a $\sbrac{0,1}$-random variable, we have:
    \begin{align*}
        \E\tbrac{X} &= \E\tbrac{\frac{1}{\vertexsamplesize} \sum_{i \in \vertexsamplesize} X_i} = \frac{1}{\vertexsamplesize} \sum_{i\in\tbrac{\vertexsamplesize}} \E\tbrac{X_i}= \frac{1}{\vertexsamplesize} \sum_{i\in\tbrac{\vertexsamplesize}} \Prob\tbrac{X_i = 1} = \frac{1}{\vertexsamplesize} \sum_{i\in\tbrac{\vertexsamplesize}} \frac{\ccdh(\degree{})}{n} = \frac{\ccdh(\degree{})}{n}
    \end{align*}
    Consequently, for all degrees $\degree{}$ such that $\ccdh({\degree{})} \geq \ccdhthreshold$,
    \begin{align*}
        \E\tbrac{\approxccdh(\degree{})} = \E\tbrac{\vertexcount X} = \vertexcount \E\tbrac{X} = \ccdh(\degree{})
    \end{align*}
    Now, we use a multiplicative Chernoff bound to obtain:
    \begin{align*}
        \Prob\tbrac{\abs{\approxccdh(\degree{}) - \ccdh(\degree{})} \geq \approxran\ccdh(\degree{})}
    \leq2\exp\fbrac{-\frac{\vertexsamplesize\ccdhthreshold\approxran^2}{3\vertexcount}}\leq\frac{1}{\vertexcount^{\constant}}
    \end{align*}
    Using a union bound over all possible such degrees gives us the statement.
\end{proof}

\subsection{Proof of~\texorpdfstring{\Cref{Lemma: General Algorithm High Degree Vertex Guarantee}}{4}}

Next, we prove~\Cref{Lemma: General Algorithm High Degree Vertex Guarantee}. 

\tailest*

The following results are used to prove~\Cref{Lemma: General Algorithm High Degree Vertex Guarantee}.

\begin{lemma}\label{Lemma: General Algorithm ApproxDegree High Degree Vertex Guarantee}
    For all vertices $\vertex$ such that $\degree{\vertex} > \degreethreshold$, Algorithm~\ref{Algorithm: Generalized CCDH Approx} ensures $\approxdegree{\vertex} \in \tbrac{\fbrac{1-\nicefrac{\approxdom}{2}}\degree{\vertex},\fbrac{1+\nicefrac{\approxdom}{2}}\degree{\vertex}}$ with high probability.
\end{lemma}

\begin{proof}
    For any vertex $\vertex \in V$, and for each edge in $\edgesample$, we assign a random variable $Y_i$, $i \in \tbrac{\edgesamplesize}$ that takes value $1$ if the $i$-th edge is incident on $\vertex$, and $0$ otherwise. Then, we have for all vertices $\vertex$ such that $\degree{\vertex} \geq \degreethreshold$, 
    $$\Prob\tbrac{Y_i = 1} = \frac{\degree{\vertex}}{\edgecount} \geq \frac{\degreethreshold}{\edgecount}$$
    We denote $Y = \nicefrac{1}{\edgesamplesize} \sum_{i \in \tbrac{\edgesamplesize}} Y_i$. Note that by line~\ref{Line: General Algorithm Approx Degree}, for all vertices $\vertex \in \vertexset$, we have $\approxdegree{\vertex} = \edgecount Y$. Then, by linearity of expectation and the fact that $Y_i$ is a $\sbrac{0,1}$-random variable, we have:
    \begin{align*}
        \E\tbrac{Y} &= \E\tbrac{\frac{1}{\edgesamplesize} \sum_{i \in \tbrac{\edgesamplesize}} Y_i} = \frac{1}{\edgesamplesize} \sum_{i \in \tbrac{\edgesamplesize}} \E\tbrac{Y_i} = \frac{1}{\edgesamplesize} \sum_{i \in \tbrac{\edgesamplesize}} \Pr\tbrac{Y_i = 1} = \frac{1}{\edgesamplesize} \sum_{i \in \tbrac{\edgesamplesize}} \frac{\degree{\vertex}}{\edgecount} = \frac{\degree{\vertex}}{\edgecount}
    \end{align*}

 Consequently, for all vertices $\vertex \in \vertexset$, we have:
    \begin{align*}
        \E\tbrac{\approxdegree{\vertex}} = \E\tbrac{\edgecount Y} = \edgecount \E\tbrac{Y} = \degree{\vertex}
    \end{align*}
    Now, we use a multiplicative Chernoff bound to obtain:
    \begin{align*}
        \Prob\tbrac{\abs{\approxdegree{\vertex}- \degree{\vertex}}\geq \frac{\approxdom\degree{\vertex}}{2}}
        \leq 2\exp\fbrac{-\frac{\edgesamplesize\degreethreshold\approxdom^2}{4\edgecount}} \leq \frac{1}{\vertexcount^{\constant}}
    \end{align*}
Using a union bound over all such possible  vertices gives us the statement.
\end{proof}

\begin{lemma}\label{Lemma: General Algorithm ApproxDegree Low Degree Vertex Guarantee}
    For any degree $\degree{} > \degreethreshold$, for all vertices $\vertex$ such that 
    $\degree{\vertex} < \frac{\degree{}}{1+\nicefrac{\approxdom}{2}}$, Algorithm~\ref{Algorithm: Generalized CCDH Approx} ensures $\approxdegree{\vertex} < \degree{}$ with high probability.
\end{lemma}

\begin{proof}
    Fix any degree $\degree{} > \degreethreshold$. For any vertex $\vertex \in V$ and for each edge in $\edgesample$, we assign a random variable $Y_i$, $i \in \tbrac{\edgesamplesize}$ that takes value $1$ if the $i$-th edge is incident on $\vertex$, and $0$ otherwise. Then, we have for all vertices $\vertex$ such that $\degree{\vertex} < \frac{\degree{}}{1+\nicefrac{\approxdom}{2}}$:
    \begin{align*}
        \Prob\tbrac{Y_i = 1} = \frac{\degree{\vertex}}{\edgecount} < \frac{\degree{}}{\edgecount\fbrac{1+\nicefrac{\approxdom}{2}}}.
    \end{align*}
    
    We denote $Y = \frac{1}{\edgesamplesize} \sum_{i \in \tbrac{\edgesamplesize}} Y_i$. Note that by line~\ref{Line: General Algorithm Approx Degree}, for all vertices $\vertex \in \vertexset$, we have $\approxdegree{\vertex} = \edgecount Y$. Therefore,
    \begin{align*}
        \E\tbrac{\approxdegree{\vertex}} = \degree{\vertex} < \frac{\degree{}}{1+\nicefrac{\approxdom}{2}}.
    \end{align*}
    Thus, the event $\approxdegree{\vertex} \geq \degree{}$ is an upper-tail deviation by a multiplicative factor at least $1+\nicefrac{\approxdom}{2}$. Using Chernoff bounds, we have:
    \begin{align*}
        \Pr\tbrac{\approxdegree{\vertex}\geq \degree{}} 
        \leq \exp\fbrac{-\frac{\edgesamplesize\degree{}\approxdom^2}{c\edgecount}}
        \leq \exp\fbrac{-\frac{\edgesamplesize\degreethreshold\approxdom^2}{c\edgecount}}
        \leq \frac{1}{\vertexcount^{\constant}},
    \end{align*}
    for a sufficiently large absolute constant $c$, where the second inequality uses $\degree{}>\degreethreshold$.
     
    Using a union bound over all vertices and all such possible degrees gives us the statement.
\end{proof}

Now we prove \Cref{Lemma: General Algorithm High Degree Vertex Guarantee}.

\begin{proof}[Proof of~\Cref{Lemma: General Algorithm High Degree Vertex Guarantee}]
     We have by Line~\ref{Line: General Algorithm Degree Based Approx} of Algorithm~\ref{Algorithm: Generalized CCDH Approx}, for all degrees $\degree{} > \degreethreshold$:
     \begin{align*}
         \approxccdh(\degree{}) = \size{\sbrac{\vertex|\approxdegree{\vertex}\geq \degree{}}}. 
     \end{align*}

     First, we prove the upper bound. By Lemma~\ref{Lemma: General Algorithm ApproxDegree Low Degree Vertex Guarantee}, with high probability, every vertex $\vertex$ such that
     \begin{align*}
         \degree{\vertex} < \frac{\degree{}}{1+\nicefrac{\approxdom}{2}}
     \end{align*}
     satisfies $\approxdegree{\vertex} < \degree{}$. Therefore,
     \begin{align*}
         \sbrac{\vertex|\approxdegree{\vertex}\geq \degree{}}
         \subseteq 
         \sbrac{\vertex|\degree{\vertex}\geq \frac{\degree{}}{1+\nicefrac{\approxdom}{2}}}.
     \end{align*}
     Hence, by~\Cref{Definition: CCDH}, we have with high probability:
     \begin{align*}
         \size{\sbrac{\vertex|\approxdegree{\vertex}\geq \degree{}}} 
         \leq 
         \size{\sbrac{\vertex|\degree{\vertex}\geq \frac{\degree{}}{1+\nicefrac{\approxdom}{2}}}} 
         \leq 
         \ccdh\fbrac{\degree{}(1-\approxdom)}.
     \end{align*}

     Next, we prove the lower bound. By Lemma~\ref{Lemma: General Algorithm ApproxDegree High Degree Vertex Guarantee}, for all vertices $\vertex$ such that
     \begin{align*}
         \degree{\vertex} \geq \frac{\degree{}}{1-\nicefrac{\approxdom}{2}},
     \end{align*}
     we have with high probability:
     \begin{align*}
         \approxdegree{\vertex} \geq \fbrac{1-\nicefrac{\approxdom}{2}}\degree{\vertex} \geq \degree{}.
     \end{align*}
     Therefore,
     \begin{align*}
         \sbrac{\vertex|\degree{\vertex}\geq \frac{\degree{}}{1-\nicefrac{\approxdom}{2}}}
         \subseteq 
         \sbrac{\vertex|\approxdegree{\vertex}\geq \degree{}}.
     \end{align*}
     Hence, by~\Cref{Definition: CCDH}, we have with high probability:
     \begin{align*}
         \size{\sbrac{\vertex|\approxdegree{\vertex}\geq \degree{}}} 
         \geq 
         \size{\sbrac{\vertex|\degree{\vertex}\geq \frac{\degree{}}{1-\nicefrac{\approxdom}{2}}}} 
         \geq
         \ccdh\fbrac{\degree{}(1+\approxdom)}.
     \end{align*}

     The inequalities follow from monotonicity of the ccdh $\ccdh$ and the fact that for $\approxdom < 1$, 
     \begin{align*}
         \frac{1}{1-\nicefrac{\approxdom}{2}} \leq (1+\approxdom)
         \qquad\text{and}\qquad
         \frac{1}{1+\nicefrac{\approxdom}{2}} \geq (1-\approxdom).
     \end{align*}
     Combining these two bounds and a union bound over all such possible degrees, we have the desired result.
\end{proof}




\subsection{Proof of \texorpdfstring{~\Cref{Theorem: General CCDH Algorithm Guarantee no hindex}}{6}}

We now prove~\Cref{Theorem: General CCDH Algorithm Guarantee no hindex}. We first state a result due to Assadi and Nguyen~\cite{AssadiNguyen/Approx/2022/OptimalH-IndexAndTriangle} that estimates the h-index of a graph.

\begin{lemma}[h-index estimator~\cite{AssadiNguyen/Approx/2022/OptimalH-IndexAndTriangle}]\label{Lemma: h-index estimator assadi}
    There exists an algorithm that, given access to degree samples from a graph $\graph$, estimates an $\esthindex$ such that $(1-\approxerror)\hindex \leq \esthindex \leq (1+\approxerror)\hindex$ with probability at least $1-\confidence$ using $\bigo{\nicefrac{\vertexcount\log(1/\confidence)}{\approxerror^2\hindex}}$ degree samples.
\end{lemma}

\magalgoplus*

\begin{proof}
    We use the algorithm of Assadi and Nguyen~\cite{AssadiNguyen/Approx/2022/OptimalH-IndexAndTriangle} to obtain a $\esthindex$ with $\approxerror = \nicefrac{1}{2}$, and $\confidence = \vertexcount^{-c}$. By~\Cref{Lemma: h-index estimator assadi}, this requires $\bigot{\frac{\vertexcount}{\hindex}}$ degree samples. 

    We now use $\approxhindex = \nicefrac{\esthindex}{2}$ for the \texttt{ApproxCCDH} algorithm. With high probability, we have $\approxhindex \in [\nicefrac{\hindex}{4},\hindex]$. Applying~\Cref{Theorem: General CCDH Algorithm Guarantee} gives us the desired result.
\end{proof}

\section{Streaming Algorithm Implementation Details: Proofs for \texorpdfstring{\Cref{SubSec: Implementation in Streaming}}{4.2}}\label{AppSec: Streaming Implementation Proofs}

Here, we present the detailed algorithms and proofs for results in \Cref{SubSec: Implementation in Streaming}.

\begin{lemma}\label{Lem: Updstrm Individual Estimator}
    In \updstrm{}, for each $j \in [\medcount]$, the degree samples $\degreesample_j$ uses $\bigot{\frac{\actvercount}{\hashrange}}$ space in expectation, and for a degree $\degree{}$, $\E\tbrac{\approxccdh_i(\degree{})} = \ccdh(\degree{})$, and $\Var\tbrac{\approxccdh(\degree{}} = \hashrange\ccdh(\degree{})$
\end{lemma}

\begin{proof}
    Observe that for any degree $\degree{}$, there are $\ccdh(\degree{})$ vertices with degree at least $\degree{}$. Let $\sequence{\vertex}{\ccdh(\degree{})}$ be the set of vertices, and $\sequence{X}{\ccdh(\degree{})}$ be indicator random variables with $X_i$ denoting whether the vertex $\vertex_i$ is selected in $\degreesample$. Then
    \begin{align*}
        \Pr\tbrac{X_i = 1} = \Pr\tbrac{\hash(\vertex_i) = 1} = \frac{1}{\hashrange}
    \end{align*}
    Let us denote $X = \sum_{i \in [\ccdh(\degree{})]} X_i$. Hence, we have
    \begin{align*}
        \E[X] = \sum_{i \in [\ccdh(\degree{})]} \E[X_i] = \frac{\ccdh(\degree{})}{\hashrange}
    \end{align*}
    The space used by $\degreesample_j$ is equal to $X$ for the case of $\degree{} = 1$, i.e. $\ccdh(\degree{}) = \actvercount$. Now, we bound the variance of $X$. 
    \begin{align*}
        \Var\tbrac{X} = \sum_{i \in [\ccdh(\degree{})]} \Var(X_i) + \sum_{i \leq j \in [\ccdh(\degree{})]} \Cov(X_i,X_j) \leq \sum_{i \in [\ccdh(\degree{})]} \E[X_i^2] = \frac{\ccdh(\degree{})}{\hashrange}
    \end{align*}
    Here, the inequality is due to the fact that $\Cov(X_i,X_j) = 0$ due to the 2-wise independence of the hash function $\hash$, and $\Var[X_i] \leq \E[X_i^2]$. Now observe that $\approxccdh_i(\degree{}) = {\hashrange} X$, hence,
    \begin{align*}
        \E[\approxccdh_i(\degree{})] &= \hashrange \E[X] = \ccdh(\degree{})\\
        \text{and}\qquad \Var[\approxccdh_i(\degree{})] &= \hashrange^2\Var[X] = \hashrange \ccdh(\degree{})
    \end{align*}
\end{proof}

\begin{lemma}[\updstrm{} Guarantee]\label{Lem: Updstrm Overall Guarantee}
    \updstrm{} uses $\bigot{\frac{\actvercount}{\approxerror^2\approxhindex}}$ space in expectation, and for all $\degree{} \leq \approxhindex$,  returns an estimator $\approxccdh(\degree{})$ such that with high probability, $\approxccdh(\degree{}) \in (1\pm \approxerror)\ccdh(\degree{})$.
\end{lemma}

\begin{proof}
    By \Cref{Lem: Updstrm Individual Estimator}, each $\approxccdh_j$ estimator takes $\bigot{\frac{\actvercount}{\hashrange}}$ space and is an unbiased estimator of $\ccdh(\degree{})$ and has variance $\Var(\approxccdh_j) \leq \hashrange\ccdh(\degree{})$. Hence, by Chebyshev's inequality
    \begin{align*}
        \Pr\tbrac{\approxccdh_j(\degree{}) \notin (1\pm\approxerror)\ccdh(\degree{})} \leq \frac{\hashrange\ccdh(\degree{})}{\approxerror^2\ccdh(\degree{})^2} \leq \frac{\constant\approxhindex}{\ccdh(\degree{})} \leq \frac{1}{4}
    \end{align*}
    Here, the last inequality follows from the fact that the ccdh value $\ccdh(d)$ for any degree $\degree{} \leq \approxhindex$ satisfies that $\approxhindex \leq \ccdh(\degree{})$. By a Chernoff-based argument (median-trick), we have that
    \begin{align*}
        \Pr\tbrac{\approxccdh(\degree{}) \notin (1\pm\approxerror)\ccdh(\degree{})} \leq \frac{1}{\vertexcount^\constant}
    \end{align*}
    Furthermore, the total space used by \updstrm{} is in expectation is
    \begin{align*}
        \frac{\medcount\actvercount}{\hashrange} = \bigot{\frac{\actvercount}{\approxerror^2\approxhindex}} \qedhere
    \end{align*}
\end{proof}

\begin{algorithm}[ht!]
    \caption{\strmccdhalg{}}\label{Alg: Stream CCDH Final}
    \begin{algorithmic}[1]
        \REQUIRE $\vertexcount,\edgecount,\approxhindex \leq \hindex,\approxdom,\approxran$
        \FUNCTION{Initialize}{}
           \STATE $\hashrange \gets \constant\approxerror^2\approxhindex$
            \STATE $\medcount \gets \constant\log\vertexcount$
            \STATE $\edgesamplesize \gets \frac{\constant\edgecount}{\approxhindex\approxdom^2}\log\vertexcount$
            \STATE For each $j \in [\medcount]$, $\func{\hash_j}{[\vertexcount]}{\hashrange} \gets$ A ${2}$-wise independent hash function.
            \STATE $\degreesample_j \gets \emptyset \forall j \in [\medcount]$ 
            \STATE $\edgesample \gets \edgesamplesize$ reservoir samplers of size $1$
        \ENDFUNCTION
        \FUNCTION{Update($\edge = \fbrac{\altvertex,\vertex}$)}
            \FOR{$j \in [{\medcount}]$}
                \FOR{$i \in \sbrac{\altvertex,\vertex}$}
                    \IF{$\hash_j(i) = 1$}
                        \IF{$ i \in \degreesample_j$}
                            \STATE Increase its degree count by $1$
                        \ELSE
                            \STATE Include $i$ in $\degreesample_j$ with degree count $1$
                        \ENDIF
                    \ENDIF
                \ENDFOR
            \ENDFOR
            \STATE Update reservoir samplers $\edgesample$
        \ENDFUNCTION
        \FUNCTION{Estimate}{}
            \STATE $\approxdegree{\vertex} = \frac{\edgecount\size{\sbrac{\edge \in \edgesample|\vertex \in \edge}}}{\edgesamplesize}, \forall \vertex \in \vertexset$
            \IF{$\degree{} \leq \approxhindex$}
                \FOR{$j \in [{\medcount}]$}
                    \STATE $\approxccdh_j(\degree{}) \gets  \hashrange \size{\sbrac{\degree{i} \in \degreesample|\degree{i} \geq \degree{}}}$
                \ENDFOR 
                \STATE $\approxccdh(\degree{}) \gets \text{ median of } \sbrac{\approxccdh_i(\degree{})|i \in [\medcount]}$
            \ENDIF
            \IF{$\degree{} > \approxhindex$}
                \STATE $\approxccdh(\degree{}) \gets \size{\sbrac{\vertex|\approxdegree{\vertex}\geq \degree{}}}$
            \ENDIF
        \ENDFUNCTION
    \end{algorithmic}
\end{algorithm}

\strmccdhfinal*

\begin{proof}
    The result follows from combining the head estimator of \updstrm{} (\Cref{Lem: Updstrm Overall Guarantee}) and the tail estimator of \streamccdh{} (\Cref{Theorem: Streaming CCDH Algorithm Guarantee}).
\end{proof}

\section{Query Algorithm Implementation Details: Proofs for \texorpdfstring{\Cref{Subsection: Query Algorithm}}{Section 4.3}}\label{AppSubsec: Query Algo Implementation}

In this section, we provide the algorithms and proofs for \Cref{Subsection: Query Algorithm}. We first describe the algorithm \updccdh{} and analyze it.

\begin{algorithm}[ht!]
    \caption{\updccdh}\label{Algorithm: Generalized CCDH Approx Updated}
    \begin{algorithmic}[1]
        \REQUIRE $\edgecount,\approxhindex$
        \STATE  \hypertarget{line1}{}\blue{$\vertexset_{\vertexsamplesize'} \gets$ $\vertexsamplesize' = (\nicefrac{\constant\edgecount}{\approxhindex\approxran^2}) \log \vertexcount$ calls to \actver{}}
        \STATE $D_{\vertexsamplesize'} \gets $ Degree of vertices in $\vertexset_{\vertexsamplesize'}$
        \STATE $\edgesample \gets$ Sample $\edgesamplesize = \frac{\constant\edgecount}{\approxhindex\approxdom^2}\log\vertexcount$ edges u.a.r. 
        \STATE $\approxdegree{\vertex} = \frac{\edgecount\size{\sbrac{\edge \in \edgesample|\vertex \in \edge}}}{\edgesamplesize}, \forall \vertex \in \vertexset$
        \IF{$\degree{} \leq \approxhindex$}
            \STATE \hypertarget{line6}{}\blue{$\approxccdh(\degree{}) \gets \frac{2\edgecount \cdot \left|\left\{ \vertex \in \vertexset_{\vertexsamplesize'} \mid \degree{\vertex} \geq \degree{} \right\}\right|}{\vertexsamplesize'} = \frac{2\edgecount\size{\sbrac{\degree{i} \in D_{\vertexsamplesize'}|\degree{i} \geq \degree{}}}}{\vertexsamplesize'}$}
        \ENDIF
        \IF{$\degree{} > \approxhindex$}
            \STATE $\approxccdh(\degree{}) \gets \size{\sbrac{\vertex|\approxdegree{\vertex}\geq \degree{}}}$
        \ENDIF
    \end{algorithmic}
\end{algorithm}

The following lemma is  analogous to \Cref{Lemma: General Algorithm High CCDH degree Guarantee}.

\begin{restatable}[Head Estimator Guarantee - Updated]{lemma}{headupd}\label{Lemma: General Algorithm High CCDH degree Guarantee-modified}
    For all degrees $\degree{}$ such that $\ccdh(\degree{}) \geq \ccdhthreshold$, \updccdh{} returns $\approxccdh(\degree{})$ such that $\approxccdh(\degree{}) \in [\fbrac{1-\approxran}\ccdh(\degree{}),$ $\fbrac{1+\approxran}\ccdh(\degree{})]$ with high probability.
\end{restatable}

\begin{proof}
     Let us focus the discussion on a particular degree $\degree{}$ such that $\degree{}\neq 0$. Note that the modified algorithm invokes \actver{} $q'=\frac{\constant\edgecount}{\ccdhthreshold\approxran^2} \log \vertexcount$ times. For each $i \in [q']$, let  $X_i$ denote a random variable that takes value $1$ if the $i$-th invocation to  Active Vertex Sampling returns a vertex with degree at least $d$ and $0$ otherwise.  Due to \Cref{Lemma: Active Vertex Sampling Probability}, any active vertex can be the output of any call to Active Vertex Sampling  with probability $1/2m$. Now, using the fact that there are $\ccdh(\degree{})$ vertices of degree at least $\degree{}$, we have:
     \begin{align*}  
         \Prob\tbrac{X_i = 1} = \frac{\ccdh(\degree{})}{2\edgecount} \geq \frac{\ccdhthreshold}{2\edgecount}
     \end{align*}
     
     We denote $X = \frac{1}{\vertexsamplesize'} \sum_{i \in \vertexsamplesize'} X_i$. By linearity of expectation, $\E\tbrac{X} = \frac{C(d)}{2m}$.

 Since we set $\approxccdh(\degree{})=\frac{2\edgecount \cdot \left|\left\{ \vertex \in \vertexsample \mid \degree{\vertex} \geq \degree{} \right\}\right|}{\vertexsamplesize'}$  where $V_{\vertexsamplesize '}$ is the multiset of vertices we get from invocations from \actver{}, we have $\approxccdh(\degree{})=2m\cdot X$.
    \begin{align*}
        \E\tbrac{\approxccdh(\degree{})} = \E\tbrac{2\edgecount \cdot X} = 2\edgecount \cdot  \E\tbrac{X} = \ccdh(\degree{}).
    \end{align*}
    Now, as in~\Cref{Lemma: General Algorithm High CCDH degree Guarantee}, we use the multiplicative Chernoff bound and apply the union bound over all possible degrees to get the desired statement.
\end{proof}

Next, we prove \Cref{Theorem: Updated CCDH Algorithm Guarantee}.

\updmagalgo*

\begin{proof}
    The result follows by combining \Cref{Lemma: General Algorithm High CCDH degree Guarantee-modified} and \Cref{Lemma: General Algorithm High Degree Vertex Guarantee} as in the proof of \Cref{Theorem: General CCDH Algorithm Guarantee}. 
\end{proof}

\begin{algorithm}[ht!]
    \caption{\updqueryccdh{}}\label{Algorithm: Query CCDH Approx Updated}
    \begin{algorithmic}[1]
        \REQUIRE $\vertexcount,\edgecount,\approxhindex  \leq \hindex,\approxdom,\approxran$
        \STATE $\vertexsamplesize' \gets \frac{\constant\edgecount}{\approxhindex\approxran^2}\log\vertexcount$
        \STATE $\edgesamplesize \gets \frac{\constant\edgecount}{\approxhindex\approxdom^2}\log\vertexcount$
        \STATE $D_{\Bar{\vertexsamplesize}} \gets \emptyset$
        \STATE $\edgesample \gets \emptyset$
        \STATE \textbackslash\textbackslash Making Queries
        \FOR{$i \in \tbrac{\vertexsamplesize'}$}
            \STATE $\edge_i \gets \randedgeq{()}$
            \STATE $\vertex_i \gets$ endpoint of $\edge_i$ chosen uniformly at random 
            \STATE $\degree{\vertex_i} \gets \degreeq\fbrac{\vertex_i}$
            \STATE $D_{\Bar{\vertexsamplesize}} \gets $ $D_{\Bar{\vertexsamplesize}} \cup \sbrac{\degree{\vertex_i}}$ with probability $\frac{1}{\degree{\vertex_i}}$
        \ENDFOR
        \FOR{$i \in \tbrac{\edgesamplesize}$}
            \STATE $\edgesample \gets \edgesample \cup \sbrac{\randedgeq{()}}$
        \ENDFOR
        \STATE \textbackslash\textbackslash Outputting the estimate
        \STATE $\approxdegree{\vertex} = \frac{\edgecount\size{\sbrac{\edge \in \edgesample|\vertex \in \edge}}}{\edgesamplesize}, \forall \vertex \in \vertexset$
        \IF{$\degree{} \leq \approxhindex$}
            \STATE $\approxccdh(\degree{}) \gets \frac{2\edgecount\size{\sbrac{\degree{i} \in D_{\Bar{\vertexsamplesize}}|\degree{i} \geq \degree{}}}}{\vertexsamplesize'}$
        \ENDIF
        \IF{$\degree{} > \approxhindex$}
            \STATE $\approxccdh(\degree{}) \gets \size{\sbrac{\vertex|\approxdegree{\vertex}\geq \degree{}}}$
        \ENDIF
    \end{algorithmic}
\end{algorithm}

Now, we prove \Cref{Theorem: Query Guarantee Updated}.

\qryupdthm*

\begin{proof}
    By design, \updqueryccdh{} uses $\bigot{\nicefrac{\edgecount}{\hindex\approxran^2}}$ \degreeq{} queries, and uses $\bigot{\nicefrac{\edgecount}{\approxhindex\approxdom^2}}$ \randedgeq{} queries. Combined, the algorithm uses $\bigot{\nicefrac{\edgecount}{\approxhindex\approxran^2}+\nicefrac{\edgecount}{\approxhindex\approxdom^2}}$ queries.
    The approximation guarantee again follows directly from Theorem~\ref{Theorem: Updated CCDH Algorithm Guarantee}.
\end{proof}

\section{Proofs for \texorpdfstring{\Cref{Section: Lower Bounds}}{Section 5}}\label{App: Proofs Lower Bounds}

In this section, we provide the proofs for \Cref{Section: Lower Bounds}. We describe the communication complexity framework in \Cref{AppSubsection: Communication General LB}, the proofs for \Cref{Subsection: h-index Lower Bound} are presented in in \Cref{AppSubSec: hindex LB}.

\subsection{Communication Problem}\label{AppSubsection: Communication General LB}

In this section, we introduce communication complexity briefly, with a particular focus on the set-disjointness problem. We show our lower bounds in both streaming and query model using reductions from set-disjointness.

In the two-party communication complexity setting, there are two parties: Alice and Bob. They wish to compute a function $\func{f}{\sbrac{0,1}^\stringlength\times\sbrac{0,1}^\stringlength}{\sbrac{0,1}}$. Alice is given an input string $\alicestring \in \sbrac{0,1}^\stringlength$. Bob is given $\bobstring \in \sbrac{0,1}^\stringlength$. While the parties know the function $f$, Alice has no knowledge of $\bobstring$, and Bob has no knowledge of $\alicestring$. They communicate bits according to a pre-specified protocol. The goal is to compute the value $f(\alicestring, \bobstring)$. In the randomized setting, the parties have access to a shared source of randomness. This is a common random string of arbitrary length which is often referred to as public randomness. A randomized protocol is said to compute $f$ if, for every input pair $(\alicestring, \bobstring)$, the correct value $f(\alicestring, \bobstring)$ is the output with probability at least $2/3$. The (randomized) communication complexity of the function $f$ is the minimum number of bits exchanged in the worst case, over all inputs, by any such (randomized) protocol that computes $f$. This is over all randomized protocols with access to public randomness that computes $f$. For more detailed exposition on communication complexity, the reader may refer to Kushilevitz and Nisan~\cite{Kushilevitz_Nisan_1996}, Rao and Yehudayoff~\cite{rao2020communication}. Now, we formally define the function of set-disjointness (denoted $\disjointness{}$):

\disjdef*

In simpler terms, the strings denote membership of elements in a set. Given $\stringlength$ elements, $\alicestring_i = 1$ denotes the presence of the $i$-th element in the set. $\disjointness{}$ evaluates to $1$ if the corresponding sets of $\alicestring$ and $\bobstring$ are disjoint, and evaluates to $0$ otherwise.


Recall that we consider the promise version of disjointness, $\promisedisjointness{\stringlength}$ where we are promised $\sum_i \alicestring_i\bobstring_i \in \sbrac{0,1}$. The following lemma quantifies the communication complexity of $\promisedisjointness{\stringlength}$.


\begin{lemma}[Hardness of $\promisedisjointness{\stringlength}$~\cite{Kushilevitz_Nisan_1996}]\label{Lemma: Hardness of Promise Disjointness}
    The communication complexity of $\promisedisjointness{\stringlength}$ is $\bigomega{\stringlength}$.
\end{lemma}


\subsection{Proofs for \texorpdfstring{~\Cref{Subsection: h-index Lower Bound}}{Section 5.2}}\label{AppSubSec: hindex LB}

In this section, we provide the proofs for \Cref{Subsection: h-index Lower Bound}. We recall the lower bound construction in \Cref{Figure: Gadget for Lower Bound h index large}.

\begin{figure}[ht!]
\centering
    \begin{tikzpicture}[scale=0.70]
        \node[shape=circle,draw=black,minimum height=3.5em] (A) at (0,0) {$A, \stringlength$};
        \node[shape=circle,draw=black,minimum height=3.5em] (B) at (8,4) {$B, \stringlength$};
        \node[shape=circle,draw=black,minimum height=3.5em] (C) at (8,-4) {$C, \stringlength$};
        \node[shape=circle,draw=black,minimum height=3.5em] (D) at (16,4) {$D, \stringlength$};
        \node[shape=circle,draw=black,minimum height=3.5em] (E) at (16,-4) {$E, \stringlength$};
        \path [-] (A) edge node[midway,above, rotate=23] {$x_i = 1$, $\frac{3\hindex}{4}$ edges} node[midway,below, rotate=23] {$x_i = 0$, $\frac{\hindex}{4}$ edges} (B);
        \path [-](A) edge node[midway,above, rotate=-23] {$y_i = 1$, $\frac{3\hindex}{4}$ edges} node[midway,below,, rotate=-23] {$y_i = 0$, $\frac{\hindex}{4}$ edges} (C);
        \path [-](B) edge node[midway,below] {$x_i = 0$, $\frac{3\hindex}{4}$ edges} node[midway,above] {$x_i = 1$, $\frac{\hindex}{4}$ edges}(D);
        \path [-](C) edge node[midway,below] {$y_i = 0$, $\frac{3\hindex}{4}$ edges} node[midway,above] {$y_i = 1$, $\frac{\hindex}{4}$ edges} (E);
    \end{tikzpicture}
    \caption{Gadget for h-index sensitive Lower Bound}\label{Figure: Gadget for Lower Bound h index large}
\end{figure}

We start with~\Cref{Theorem: Streaming h-index Lower Bound}.

\StrmhindLB*

\begin{proof}
    We use the construction described above to reduce from the $\promisedisjointness{\stringlength}$ problem. We construct a graph on $\vertexcount$ vertices with six independent sets of vertices $\vertexsetA,\vertexsetB,\vertexsetC,\vertexsetD,\vertexsetE,\isovertexset$ with $\size{\vertexsetA} = \size{\vertexsetB} = \size{\vertexsetC} = \size{\vertexsetD} = \size{\vertexsetE} = \stringlength$, and $\size{\isovertexset} = \vertexcount - 5\stringlength$. For the stream of edges, given an input $\alicestring \in \sbrac{0,1}^\stringlength$, Alice generates the edge stream of $\stream^A = \fbrac{\stream^A_1,\stream^A_2,\ldots,\stream^A_\stringlength}$ where for each $i \in \tbrac{\stringlength}$, Alice generates $\hindex$ edges $\stream^A_i$ as follows:
    \begin{align*}
        \fbrac{\fbrac{\vertexA_i,\vertexB_i},\ldots,\fbrac{\vertexA_i,\vertexB_{i+\frac{\hindex}{4}-1}},\fbrac{\vertexD_i,\vertexB_{i+\frac{\hindex}{4}}},\ldots,\fbrac{\vertexD_i,\vertexB_{i+\hindex-1}}}\text{ ,if } \alicestring_i = 0\\
        \fbrac{\fbrac{\vertexA_i,\vertexB_i},\ldots,\fbrac{\vertexA_i,\vertexB_{i+\frac{3\hindex}{4}-1}},\fbrac{\vertexD_i,\vertexB_{i+\frac{3\hindex}{4}}},\ldots,\fbrac{\vertexD_i,\vertexB_{i+\hindex-1}}}\text{ ,if } \alicestring_i = 1
    \end{align*}

    Correspondingly, given an input $\bobstring \in \sbrac{0,1}^\stringlength$, Bob generates the edge stream of $\stream^B = \fbrac{\stream^B_1,\stream^B_2,\ldots,\stream^B_\stringlength}$ where for each $i \in \tbrac{\stringlength}$, Bob generates $\hindex$ edges $\stream^B_i$ as follows:

    \begin{align*}
        \fbrac{\fbrac{\vertexA_i,\vertexC_i},\ldots,\fbrac{\vertexA_i,\vertexC_{i+\frac{\hindex}{4}-1}},\fbrac{\vertexE_i,\vertexC_{i+\frac{\hindex}{4}}},\ldots,\fbrac{\vertexE_i,\vertexC_{i+\hindex-1}}}\text{ ,if } \bobstring_i = 0\\
        \fbrac{\fbrac{\vertexA_i,\vertexC_i},\ldots,\fbrac{\vertexA_i,\vertexC_{i+\frac{3\hindex}{4}-1}},\fbrac{\vertexE_i,\vertexC_{i+\frac{3\hindex}{4}}},\ldots,\fbrac{\vertexE_i,\vertexC_{i+\hindex-1}}}\text{ ,if } \bobstring_i = 1
    \end{align*}

    Now, observe that the graph has $\vertexcount$ vertices by construction, and $\edgecount = 2\stringlength\hindex$ edges. Note that each vertex in $\vertexsetB$ and $\vertexsetC$ has degree exactly $\hindex$, and there are total of $2\stringlength$ such vertices. Each vertex in $\vertexsetD$ and $\vertexsetE$ has degree at most $\frac{3\hindex}{4}$. Each vertex $\vertexA_i$ in $\vertexsetA$ has degree $\frac{6\hindex}{4}$ if $\alicestring_i = \bobstring_i = 1$, and at most $\hindex$ otherwise. As we have considered the promise version of $\disjointness{}$, there can be only one vertex in $\vertexsetA$ with degree greater than $\hindex$.  Hence, the constructed graph has h-index $\hindex$. Furthermore, if $\promisedisjointness{\stringlength}\fbrac{\alicestring,\bobstring} = 1$, then we have $\ccdh(i) > \hindex$ if $i \leq \hindex$, and $0$ otherwise. Alternatively, if $\promisedisjointness{\stringlength}\fbrac{\alicestring,\bobstring} = 0$, then $\ccdh(i) > \hindex$ if $i \leq \hindex$, $\ccdh(i) = 1$ if $\hindex < i \leq \frac{6\hindex}{4}$, and $0$ otherwise.
    
    Let $\algo$ be a constant-pass streaming algorithm that uses $\streamspace$ space to compute the $\biapprox$-BMA of the CCDH with probability at least $2/3$. As the algorithm $\algo$ computes an $\biapprox$-BMA $\approxccdh$ for $\approxdom < \frac{1}{5}$, we would have:
    \begin{align*}
        \approxccdh\fbrac{\frac{5\hindex}{4}} = \begin{dcases}
            >0 & if \text{ }\promisedisjointness{\stringlength}\fbrac{\alicestring,\bobstring} = 0\\
            0 & if \text{ }\promisedisjointness{\stringlength}\fbrac{\alicestring,\bobstring} = 1
        \end{dcases}
    \end{align*}

    Then, we would be able to solve $\promisedisjointness{\stringlength}$ with probability at least $2/3$ using $\bigo{\streamspace}$ bits. Given the $\bigomega{\stringlength}$ lower bound on $\promisedisjointness{}$ (\Cref{Lemma: Hardness of Promise Disjointness}), an argument similar to the proof of \Cref{Theorem: Streaming General LB}, together with the fact that $\stringlength = \frac{\edgecount}{2\hindex}$, implies that $\streamspace = \bigomega{\frac{\edgecount}{\hindex}}$.
\end{proof}

Next, we prove \Cref{Theorem: Query h-index Lower Bound}.

\QueryhindLB*

\begin{proof}
    We first describe the embedding $\embedding$ from the strings $\alicestring$ and $\bobstring$. The embedding $\embedding$ that we use here is as described above in Figure~\ref{Figure: Gadget for Lower Bound h index}. Formally, we start with the vertex set $\ccvertexset$, consisting of six independent sets of vertices $\vertexsetA,\vertexsetB,\vertexsetC,\vertexsetD,\vertexsetE,\isovertexset$ with $\size{\vertexsetA} = \size{\vertexsetB} = \size{\vertexsetC} = \size{\vertexsetD} = \size{\vertexsetE} = \stringlength$, and $\size{\isovertexset} = \vertexcount - 5\stringlength$. For the edge set $\ccedgeset$, for each $\alicestring_i$, we add edges:
    \begin{align*}
        &\sbrac{\fbrac{\vertexA_i,\vertexB_i},\ldots,\fbrac{\vertexA_i,\vertexB_{i+\frac{\hindex}{4}-1}},\fbrac{\vertexD_i,\vertexB_{i+\frac{\hindex}{4}}},\ldots,\fbrac{\vertexD_i,\vertexB_{i+\hindex-1}}}  \text{, if } \alicestring_i=0\\        &\sbrac{\fbrac{\vertexA_i,\vertexB_i},\ldots,\fbrac{\vertexA_i,\vertexB_{i+\frac{3\hindex}{4}-1}},\fbrac{\vertexD_i,\vertexB_{i+\frac{3\hindex}{4}}},\ldots,\fbrac{\vertexD_i,\vertexB_{i+\hindex-1}}}  \text{, if } \alicestring_i = 1       
    \end{align*}
Correspondingly, we also add  for each $\bobstring_i$:
    \begin{align*}
        &\sbrac{\fbrac{\vertexA_i,\vertexC_i},\ldots,\fbrac{\vertexA_i,\vertexC_{i+\frac{\hindex}{4}-1}},\fbrac{\vertexE_i,\vertexC_{i+\frac{\hindex}{4}}},\ldots,\fbrac{\vertexE_i,\vertexC_{i+\hindex-1}}} \text{, if }\bobstring_i = 0\\
        &\sbrac{\fbrac{\vertexA_i,\vertexC_i},\ldots,\fbrac{\vertexA_i,\vertexC_{i+\frac{3\hindex}{4}-1}},\fbrac{\vertexE_i,\vertexC_{i+\frac{3\hindex}{4}}},\ldots,\fbrac{\vertexE_i,\vertexC_{i+\hindex-1}}} \text{, if } \bobstring_i = 1
    \end{align*}

     For any $\alicestring,\bobstring \in \sbrac{0,1}^\stringlength$, $\embedding(\alicestring,\bobstring)$ creates a graph $\ccgraph = \fbrac{\ccvertexset,\ccedgeset}$ with $\vertexcount$ vertices, $2\stringlength\hindex$ edges. Fixing $\stringlength = \frac{\edgecount}{\hindex}$ makes the number of edges $\edgecount$. Again, by the earlier arguments, the h-index of the graph is always $\hindex$. Now, consider $\algo$ to be an algorithm that computes an $\biapprox$-BMA of the ccdh for all graphs using $\bigo{\ccspace}$ queries. As stated earlier, we have for this construction:
    \begin{align*}
        \approxccdh_{\ccgraph}\fbrac{\frac{5\hindex}{4}} = \begin{dcases}
            >0 & if \text{ }\promisedisjointness{\stringlength}\fbrac{\alicestring,\bobstring} = 0\\
            0 & if \text{ }\promisedisjointness{\stringlength}\fbrac{\alicestring,\bobstring} = 1
        \end{dcases}
    \end{align*}
    Thus, we can solve $\promisedisjointness{}$ using $\algo$. Now, we discuss how to implement the queries in the communication model:
    \begin{itemize}
        \item \textbf{\degreeq{}: } Given a vertex in $\vertexsetB$ or $\vertexsetC$, return $\hindex$. For a vertex $\vertexD_i \in \vertexsetD$ (resp. $\vertexE \in \vertexsetE$), return $\frac{3\hindex}{4}$ if $\alicestring_i = 0$ (resp. $\bobstring_i = 0$), and return $\frac{\hindex}{4}$ if $\alicestring_i = 1$ (resp. $\bobstring_i = 1$). For a vertex in $\vertexsetA$, if $\alicestring_i =\bobstring_i = 0$, return $\frac{\hindex}{2}$. If $\alicestring_i = 0$ and $\bobstring_i = 1$ or $\alicestring_i = 1$ and $\bobstring_i = 0$, return $\hindex$. If $\alicestring_i = \bobstring_i = 1$, return $\frac{6\hindex}{4}$. Hence, for any degree query, exchanging at most $2$ elements of the strings suffice. Thus, $\degreeq{}$ can be implemented using at most $2$ bits of communication.
        
        \item \textbf{\randedgeq{}: } Each edge in the graph is associated with an element in either of the two strings. On the other hand, each string element is associated with exactly $\hindex$ edges. To generate a $\randedgeq{}$, choose one index out of the $[2\stringlength]$ indices at random using the public randomness and communicate that bit, then use shared randomness to obtain an edge out of the possible $\hindex$ edges uniformly at random. Hence, we need exactly $1$ bit of communication to implement the $\randedgeq{}$ given access to shared random string available to Alice and Bob.
        \item \textbf{\edgeexistsq{}: } Each vertex-pair has at most one possible way of having an edge, defined by corresponding $\alicestring_i$ or $\bobstring_i$. Communicating this bit suffices to decide whether the given vertex-pair contains an edge or not. Thus, communicating $1$ bit suffices to implement $\edgeexistsq{}$ query.
        \item \textbf{\neighbourq{$\fbrac{\vertex,j}$}: } For a vertex $\vertexD_i$ (resp. $\vertexE_i$) in $\vertexsetD$ (resp. $\vertexsetE$), if corresponding $\alicestring_i$ (resp. $\bobstring_i$) is $1$, return $\vertexC_{i+j}$ if $j \leq \frac{h}{4}$, and if $\alicestring_i$ (resp. $\bobstring_i$) is $0$, return $\vertexC_{i+j}$ if $j \leq \frac{3h}{4}$. For a vertex $\vertexB_i$ (resp. $\vertexC_i$) in $\vertexsetB$ (resp. $\vertexsetC$), if $j \geq \frac{3\hindex}{4}$, return $\vertexD_{i-\hindex+j}$ (resp. $\vertexE_{i-\hindex+j}$), else if $j \leq \frac{\hindex}{4}$, return $\vertexA_{i-\hindex+j}$, else if $\alicestring_{i-\hindex+j} = 1$ (resp. $\bobstring_{i-\hindex+j} = 1$), return $\vertexA_{i-\hindex+j}$, else return $\vertexD_{i-\hindex+j}$ (resp. $\vertexE_{i-\hindex+j}$). For a vertex $\vertexA_i$ in $\vertexsetA$, we define the \neighbourq{} queries through the four cases outlined below:
        \begin{itemize}
            \item \textbf{$\alicestring_i = 0, \bobstring_i = 0$: } For $j \leq \frac{\hindex}{4}$, return $\vertexB_{i+j-1}$, for $\frac{\hindex}{4} < j \leq \frac{\hindex}{2}$, return $\vertexC_{i+\frac{\hindex}{4}+j-1}$. Else, return $\perp$.
            \item \textbf{$\alicestring_i = 0, \bobstring_i = 1$: } For $j \leq \frac{\hindex}{4}$, return $\vertexB_{i+j-1}$, for $\frac{\hindex}{4} < j \leq \hindex$, return $\vertexC_{i+\frac{\hindex}{4}+j-1}$. Else, return $\perp$.
            \item \textbf{$\alicestring_i = 1, \bobstring_i = 0$: } For $j \leq \frac{3\hindex}{4}$, return $\vertexB_{i+j-1}$, for $\frac{3\hindex}{4} < j \leq \hindex$, return $\vertexC_{i+\frac{3\hindex}{4}+j-1}$. Else, return $\perp$.
            \item \textbf{$\alicestring_i = 1, \bobstring_i = 1$: } For $j \leq \frac{3\hindex}{4}$, return $\vertexB_{i+j-1}$, for $\frac{3\hindex}{4} < j \leq \frac{6\hindex}{4}$, return $\vertexC_{i+\frac{3\hindex}{4}+j-1}$. Else, return $\perp$.
        \end{itemize}
        Thus, any \neighbourq{} query can be answered by communication of $2$ bits.
    \end{itemize}
    As described above, all the queries can be implemented using $\bigo{1}$ bits of communication. Thus, any algorithm that computes an $\biapprox$-BMA of ccdh for any $\approxdom < \frac{1}{3}$ requires to make $\bigomega{\stringlength}$ queries. Since $\edgecount = 2\stringlength\hindex$, we have the stated lower bound.
\end{proof}



\section{Empirical Evaluations}\label{AppSec: Experiments}

Given the tight bounds we have established across both models, we focus on two broad issues in our experiments: \textbf{1)} How does the algorithm work in practice? and \textbf{2)} What does the complexity parameter $\bigtheta{\nicefrac{\edgecount}{\hindex}}$ imply for real world graphs? Keeping these two broad issues in mind, we focus on the following specific questions in our experiments:

\noindent\textbf{Q1:} How does the error vary across different parts (i.e., different values of degree) of the ccdh?

\noindent\textbf{Q2:} How does the error vary with different sampling budgets?

\noindent\textbf{Q3:} How does the algorithm perform for real-world graphs?

\noindent\textbf{Q4:} What does the complexity of $\nicefrac{\edgecount}{\hindex}$ imply for real-world graphs?

\subsection{Experimental Setup:} We evaluate on graph datasets from  Leskovec and Krevl~\cite{snapnets,RossiAhmed/KDD/2016/NetworkRepository} including social-network graphs, web graphs, location graphs, and other
real-world networks. All experiments were run on a 3.6 GHz 12th Gen Intel Core i7-12700 processor with 12 cores and 32 GB RAM. We provide more detailed experimental results and comparison with baseline algorithms in \Cref{AppSec: Experiments}.

We implement the algorithm in its generic form (Algorithm~\ref{Algorithm: Generalized CCDH Approx}) by obtaining degree and edge samples directly. We set $\approxdom = \approxran = 0.1$. For simplicity, we use the exact value of the \(h\)-index in the experiments. To evaluate the performance of the algorithm at different query budgets, we vary the constant $\constant$ (Lines 1 and 3 in \magccdh{}). For the baseline algorithms, we match the total number of primitive sampling/query operations used by \magccdh{} at the corresponding value of $\constant$. We run the algorithms $50$ times for each graph and each value of the constant $\constant$. For each degree value $\degree{}$, we measure the estimation error as the smallest value of $\approxerror$ such that $(1-\approxerror)\ccdh((1+\approxerror)\degree{}) \leq \approxccdh(\degree{}) \leq (1+\approxerror)\ccdh((1-\approxerror)\degree{})$. We define the fraction of queries as the total number of sampled vertices used by \magccdh{} divided by $n$, expressed as a percentage.

\subsection{Regarding Q1}\label{AppSubSec: Experiments Q1}

\noindent\textbf{Q1:} How does the error vary across different parts (i.e., different values of degree) of the ccdh?

 In \Cref{fig:AppHeatMap,fig:QueryHeatMap,fig:StreamHeatMap}, we plot the heatmaps of median error across degree values at a fixed resource budget ($c = 0.01$), with each graph's h-index marked by black dashes. We evaluate the errors at degree values at exponents of $2$, starting from $1$.

For \magccdh{}, the largest errors typically occur near the $h$-index of the corresponding graph. The same pattern remains true for the streaming baseline~\cite{SeshadriMcGregor/ICDM/2015/CCDHStreamingEmpirical}, with the errors being highest before the h-index. This is due to the baseline's heuristic thresholding technique prematurely switching to the tail estimate. For the query baseline~\cite{Eden/WWW/2017/CCDHinQueryModel}, the errors increase both immediately after the $h$-index and near the extreme tail of the degree distribution. These results also show that \magccdh{} performs better than the baselines across the degree spectrum.

\begin{figure*}[ht]
    \centering
    
    \begin{minipage}[t]{0.52\textwidth}
        \centering
        
        \includegraphics[width=\linewidth]{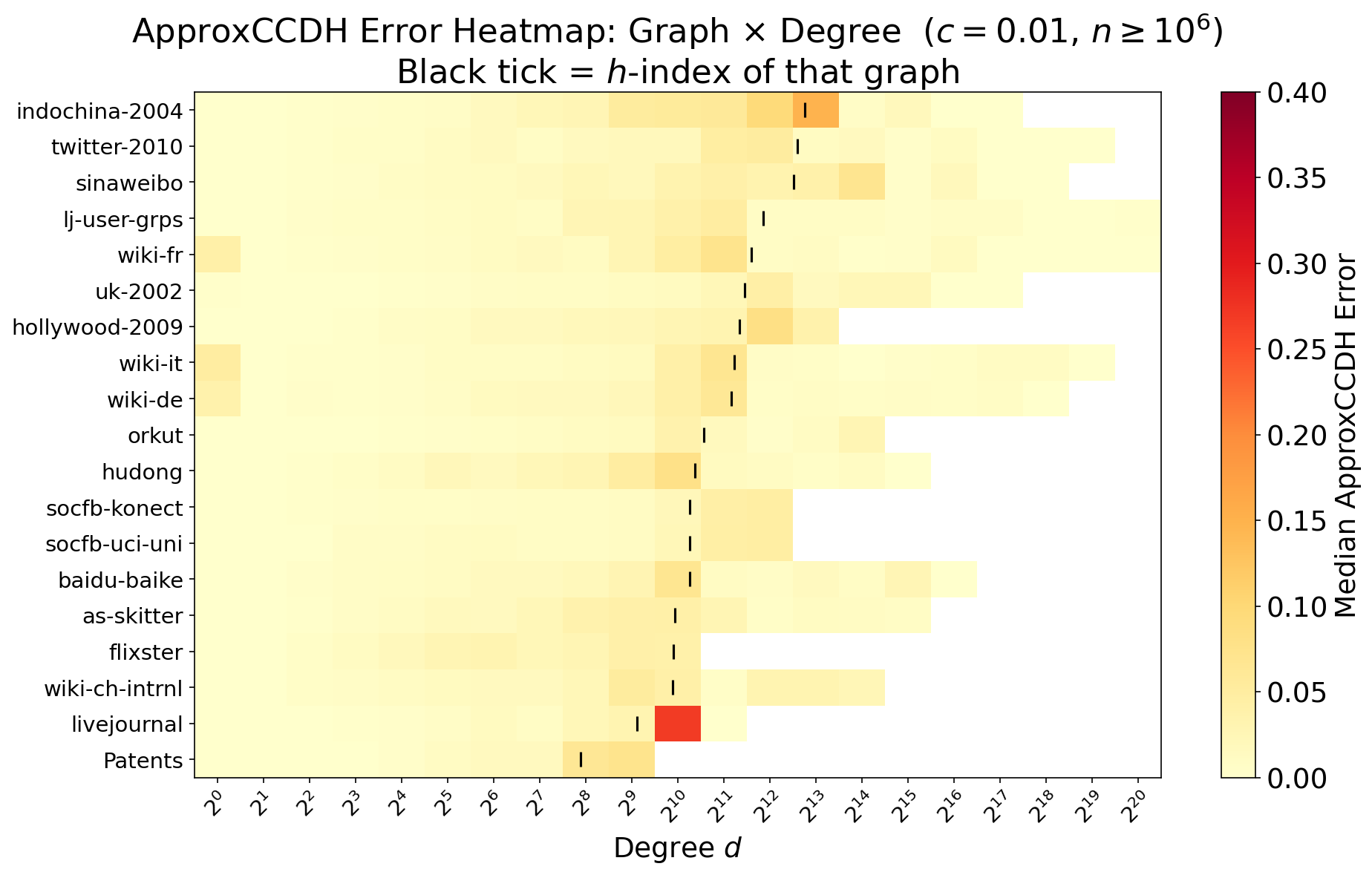}
        \caption{Error Heatmap across degrees for \magccdh{}}
        \label{fig:AppHeatMap}
        
        \vspace{0.6cm}
        
        \includegraphics[width=\linewidth]{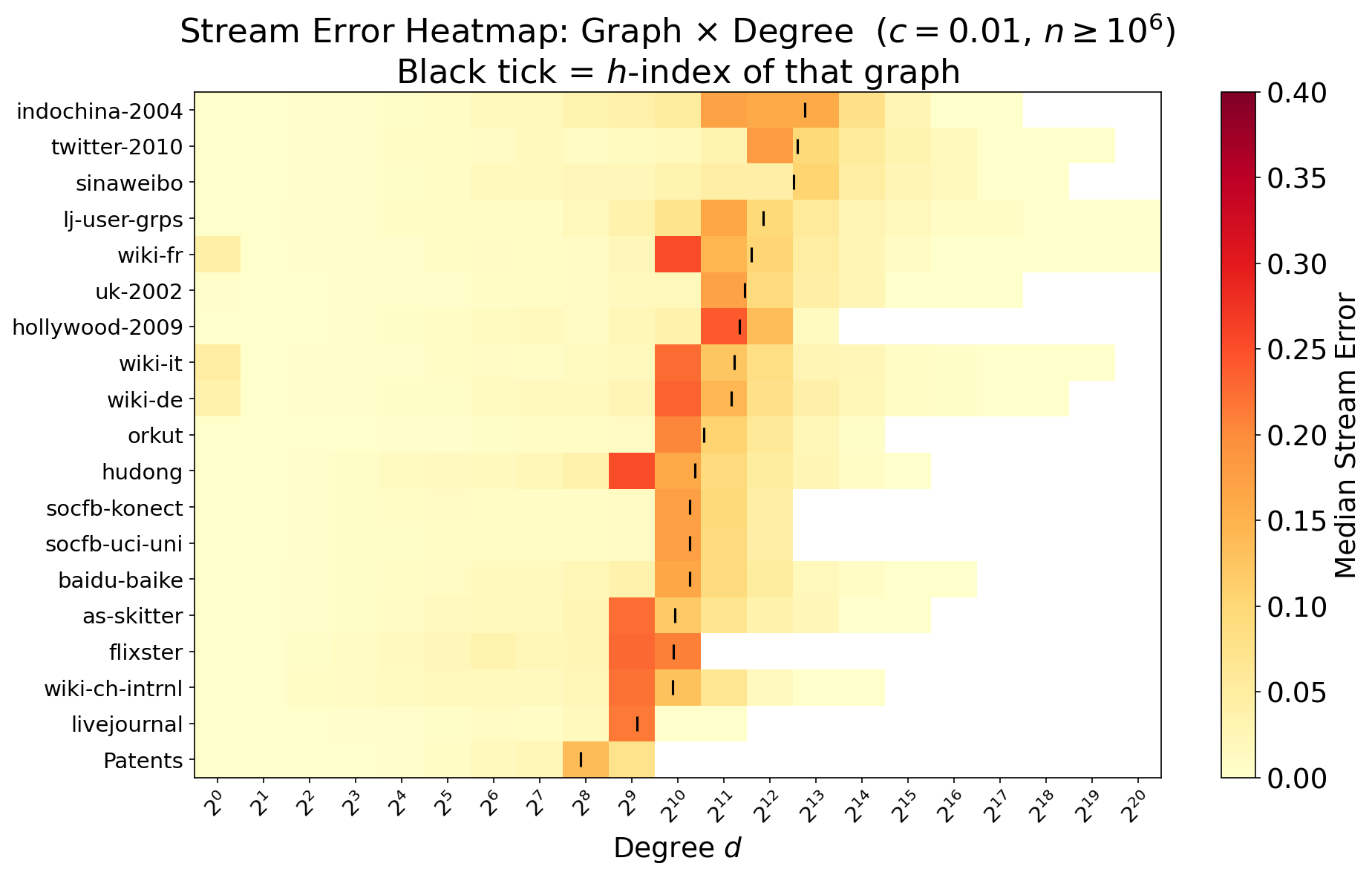}
        \caption{Error Heatmap across degrees for Streaming Baseline~\cite{SeshadriMcGregor/ICDM/2015/CCDHStreamingEmpirical}}
        \label{fig:StreamHeatMap}
        
        \vspace{0.6cm}
        
        \includegraphics[width=\linewidth]{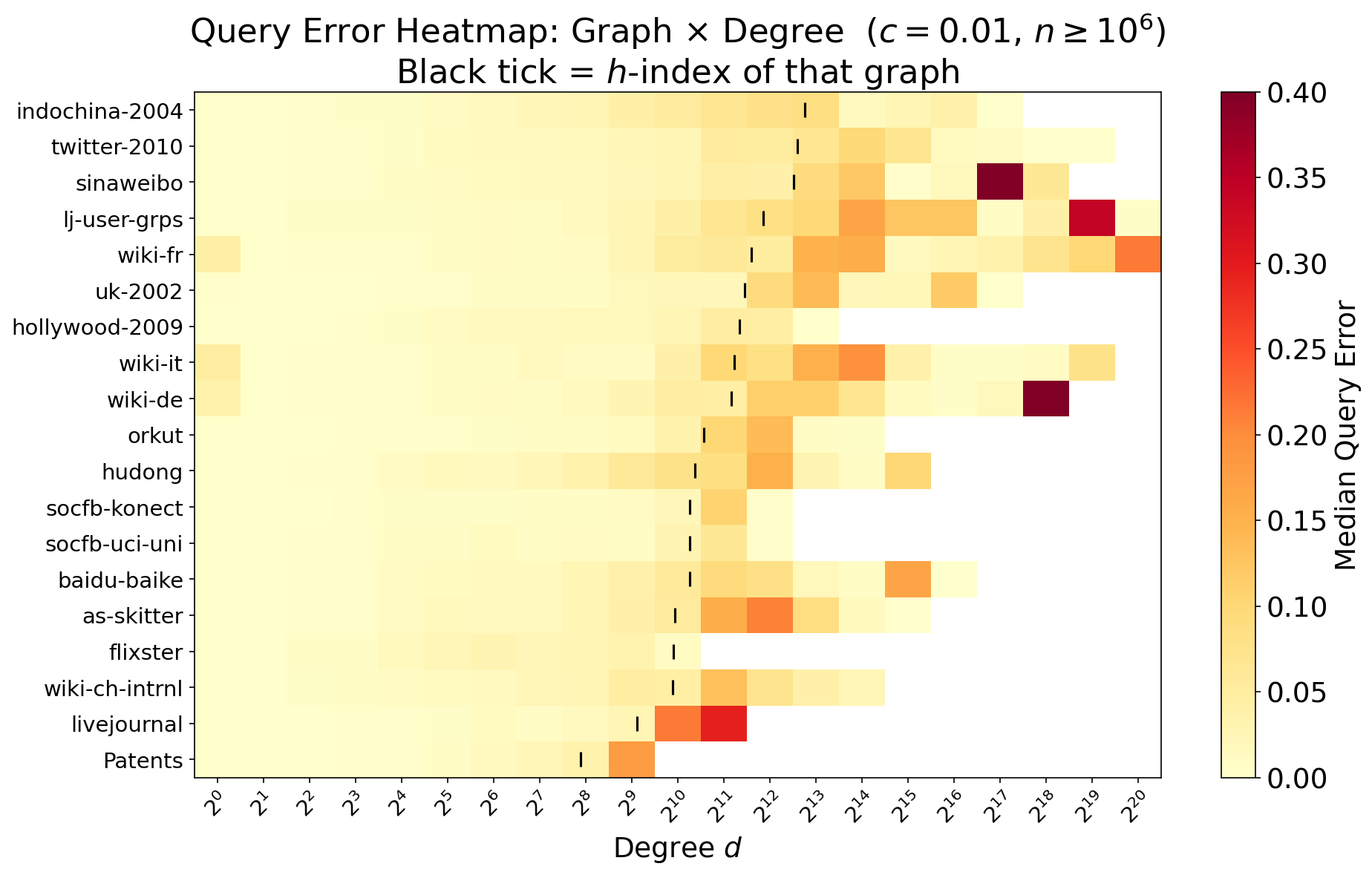}
        \caption{Error Heatmap across degrees for Query Baseline~\cite{Eden/WWW/2017/CCDHinQueryModel}}
        \label{fig:QueryHeatMap}
        
    \end{minipage}\hfill
    \begin{minipage}[t]{0.44\textwidth}
        \centering
        
        \includegraphics[width=\linewidth]{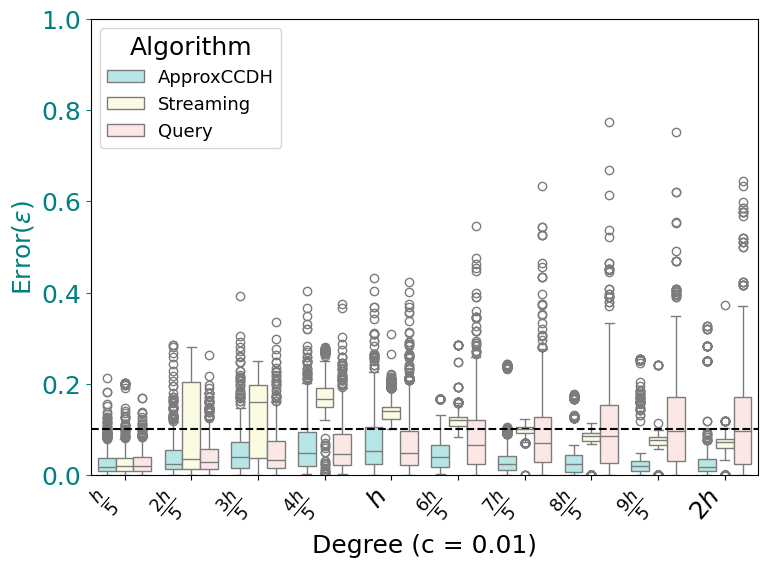}
        \caption{Error at degree values around $\hindex$ for $\constant = 0.01$}
        \label{fig:ErrhatConstantC01}
        
        \vspace{0.6cm}
        
        \includegraphics[width=\linewidth]{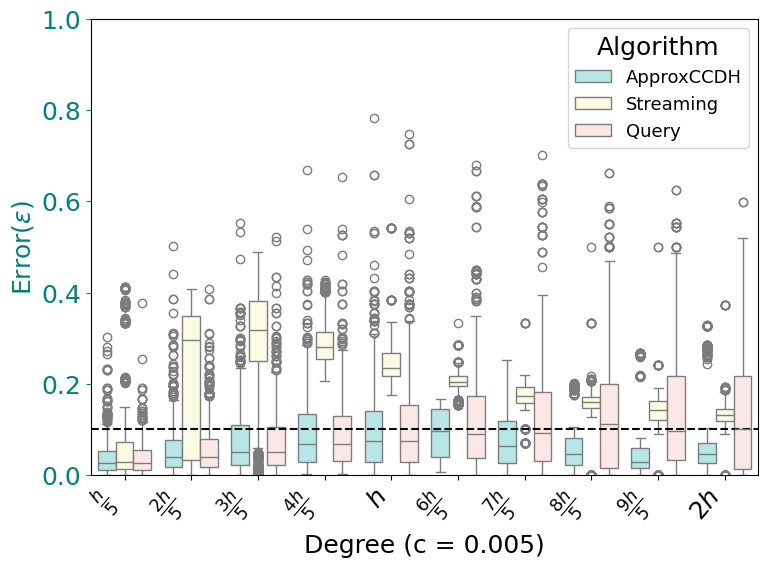}
        \caption{Error at degree values around $\hindex$ for $\constant = 0.005$}
        \label{fig:ErrhatConstantC005}
        
        \vspace{0.6cm}
        
        \includegraphics[width=\linewidth]{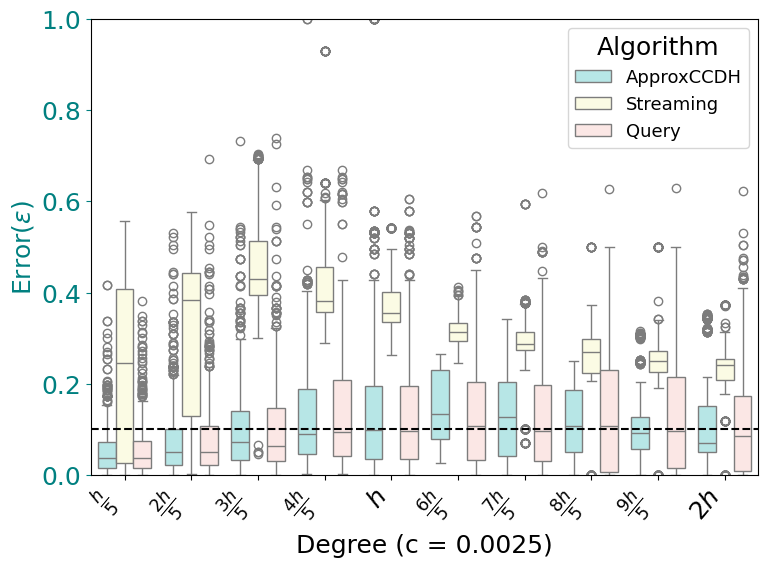}
        \caption{Error at degree values around $\hindex$ for $\constant = 0.0025$}
        \label{fig:ErrhatConstantC0025}
        
    \end{minipage}
    \caption*{Figures corresponding to \Cref{AppSubSec: Experiments Q1}}
\end{figure*}

Motivated by these observations, we take a closer look at the error around the h-index for \magccdh{} and the baselines under different query budgets. In \Cref{fig:ErrhatConstantC01,fig:ErrhatConstantC005,fig:ErrhatConstantC0025}, we evaluate the error values of all three algorithms at $[\nicefrac{\hindex}{5},\nicefrac{2\hindex}{5},\nicefrac{3\hindex}{5},\nicefrac{4\hindex}{5},\hindex,\nicefrac{6\hindex}{5},\nicefrac{7\hindex}{5},\nicefrac{8\hindex}{5},\nicefrac{9\hindex}{5},2\hindex]$ for the values of $\constant$ in $[0.01,0.005,0.0025]$. 

For the evaluated budgets and degrees, \magccdh{} remains comparable to, and often better than, both baselines across the evaluated degree values. For the smaller values of $c$, the streaming baseline has noticeably larger error across the degree values considered. The query algorithm incurs increased error for degree values larger than $\hindex$. In contrast, \magccdh{} performs consistently throughout the range of degree values.


\subsection{Regarding Q2}\label{AppSubSec: Experiments Q2}

\noindent\textbf{Q2:} How does the error vary with different sampling budgets? 


In \Cref{fig:CvsError4h5,fig:CvsErrorH,fig:CvsError6h5}, we study the errors at degree values $[\nicefrac{4\hindex}{5},\hindex,\nicefrac{6\hindex}{5}]$. We observe that the streaming algorithm of Simpson et al.~\cite{SeshadriMcGregor/ICDM/2015/CCDHStreamingEmpirical} fails to satisfy the error margin for all values of $\constant$ less than $0.1$. While the query algorithm of Eden et al.~\cite{Eden/WWW/2017/CCDHinQueryModel} remains competitive with \magccdh{} for $\nicefrac{4\hindex}{5}$, and $\hindex$, they incur increased error at $\nicefrac{6\hindex}{5}$. Overall, \magccdh{} remains within the desired error margin at $\constant = 0.01$.

In \Cref{fig:CvsError2dminh,fig:CvsError3dminh,fig:CvsError4dminh}, we study the error values at $[\nicefrac{(2\degree{\max}+3\hindex)}{5},\nicefrac{(3\degree{\max}+2\hindex)}{5},\nicefrac{(4\degree{\max}+\hindex)}{5}]$. For all values of $\constant$, \magccdh{} and the streaming baseline algorithm~\cite{SeshadriMcGregor/ICDM/2015/CCDHStreamingEmpirical} have error below $0.1$. The query algorithm of Eden et al.~\cite{Eden/WWW/2017/CCDHinQueryModel} incurs more error as the degree value increases. At $\degree{} = \nicefrac{(4\degree{\max}+\hindex)}{5}$, their algorithm has high error for all values of $\constant$ considered.

\begin{figure*}[ht!]
    \centering
    
    \begin{minipage}[t]{0.48\textwidth}
        \centering
        
        \includegraphics[width=\linewidth]{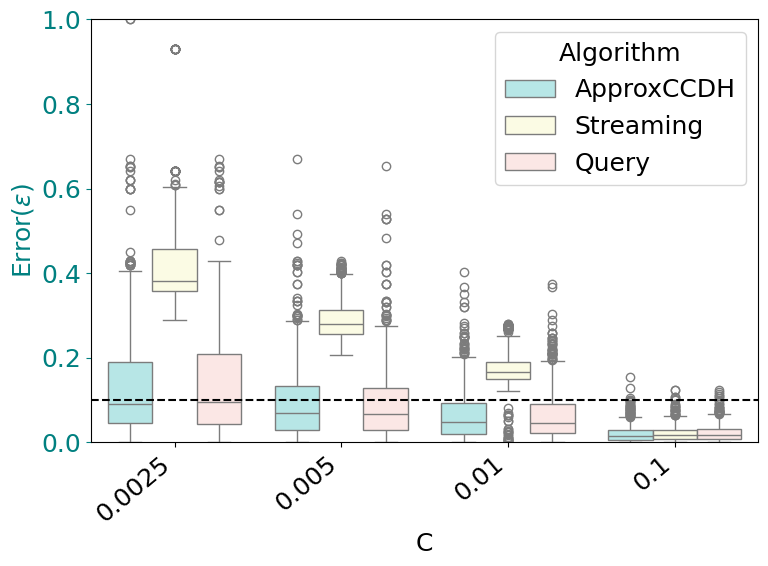}
        \caption{Error at $\nicefrac{4\hindex}{5}$ across different values of the constant $\constant$.}
        \label{fig:CvsError4h5}
        
        \vspace{0.6cm}
        
        \includegraphics[width=\linewidth]{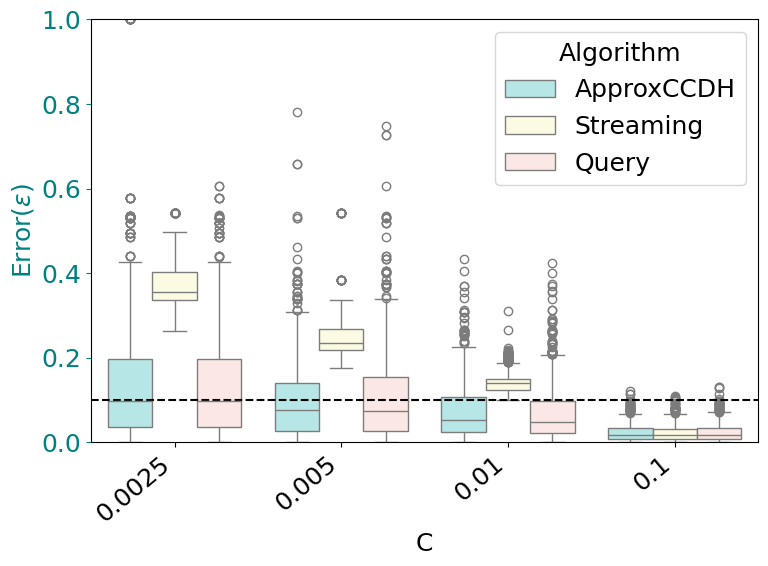}
        \caption{Error at $\hindex$ across different values of the constant $\constant$.}
        \label{fig:CvsErrorH}
        
        \vspace{0.6cm}
        
        \includegraphics[width=\linewidth]{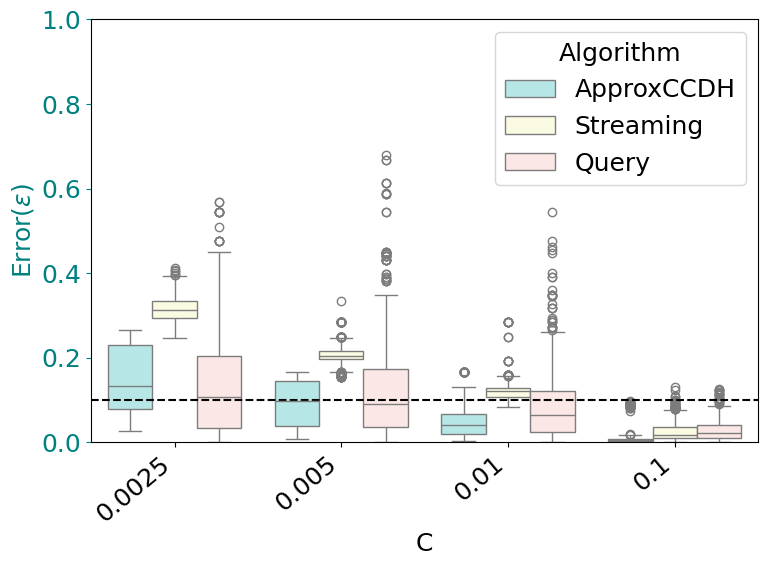}
        \caption{Error at $\nicefrac{6\hindex}{5}$ across different values of the constant $\constant$.}
        \label{fig:CvsError6h5}
        
    \end{minipage}\hfill
    \begin{minipage}[t]{0.48\textwidth}
        \centering
        
        \includegraphics[width=\linewidth]{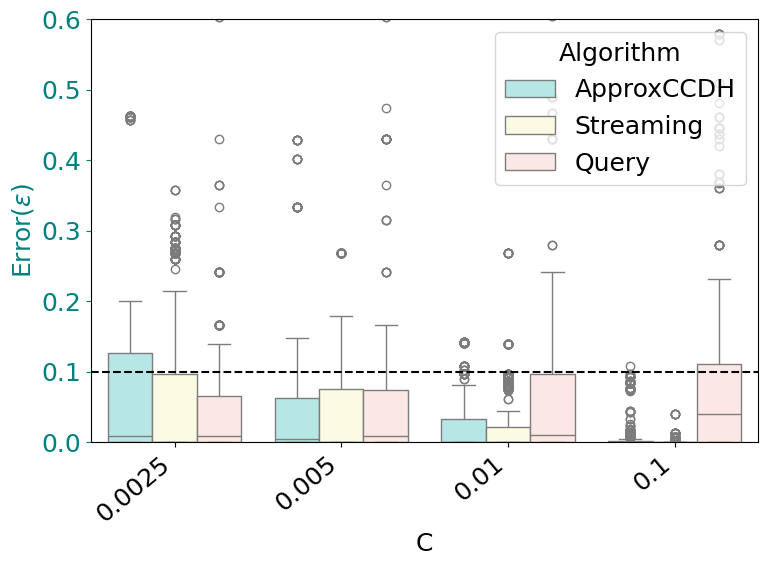}
        \caption{Error at $\nicefrac{2\degree{\max}+3\hindex}{5}$ across different values of the constant $\constant$.}
        \label{fig:CvsError2dminh}
        
        \vspace{0.6cm}
        
        \includegraphics[width=\linewidth]{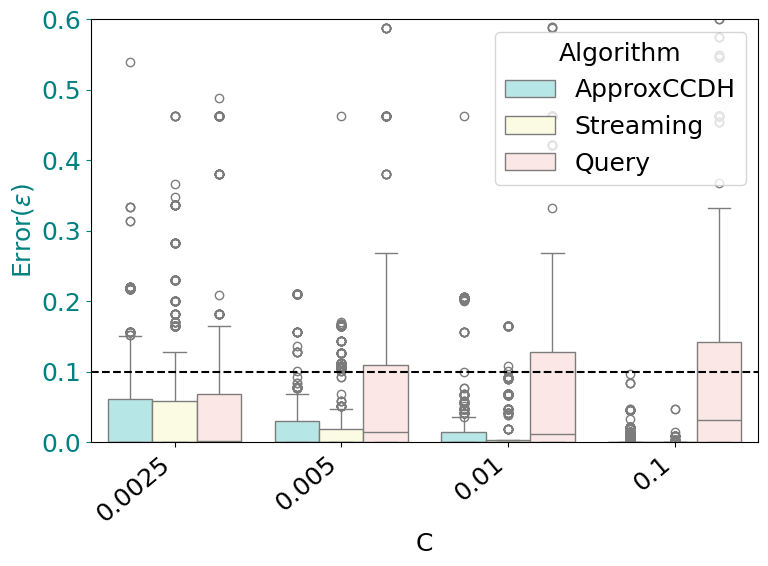}
        \caption{Error at $\nicefrac{3\degree{\max}+2\hindex}{5}$ across different values of the constant $\constant$.}
        \label{fig:CvsError3dminh}
        
        \vspace{0.6cm}
        
        \includegraphics[width=\linewidth]{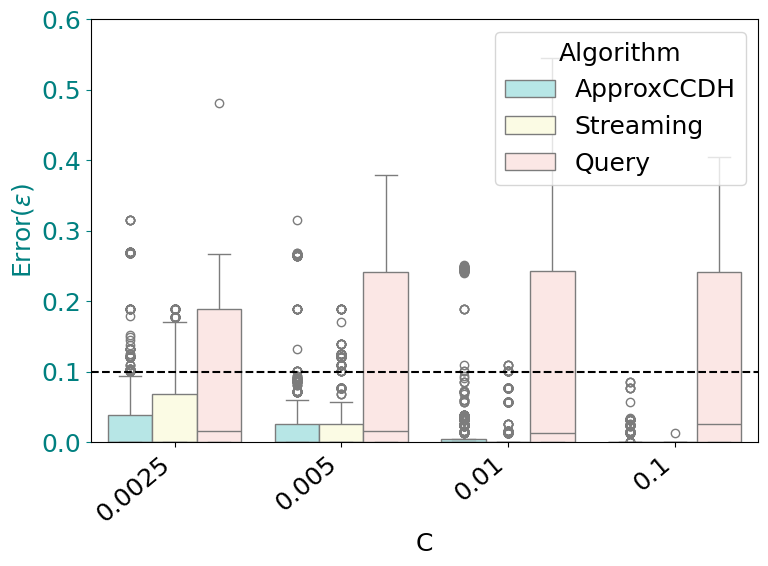}
        \caption{Error at $\nicefrac{4\degree{\max}+\hindex}{5}$ across different values of the constant $\constant$.}
        \label{fig:CvsError4dminh}
        
    \end{minipage}
    \caption*{Figures corresponding to \Cref{AppSubSec: Experiments Q2}}
\end{figure*}


\subsection{Regarding Q3}\label{AppSubSec: Experiments Q3}

\noindent\textbf{Q3:} How does the algorithm perform for real-world graphs?


In \Cref{fig:P1_FullErrorCurve}, we plot the median error across all graphs and trials ($50$ trials per graph) for various graph types. \Cref{fig:P8_WebSmallMultiples,fig:P9_SocSmallMultiples} plot the median error across $50$ trials for individual web graphs and social-network graphs, respectively.

Across the evaluated graphs, \magccdh{} is competitive with the baselines across the degree spectrum. The streaming baseline algorithm~\cite{SeshadriMcGregor/ICDM/2015/CCDHStreamingEmpirical} consistently shows increased error just before the h-index of the graph. As discussed above, this is due to their heuristic thresholding technique. The query baseline algorithm also incurs increased error immediately following the h-index, and toward the tail end of the degree distribution. In contrast, \magccdh{} maintains error below $0.1$ throughout the degree spectrum for almost all graphs.

\begin{figure*}[t]
    \centering
    \begin{minipage}[b]{0.99\textwidth}
        \centering
        \includegraphics[width=\linewidth]{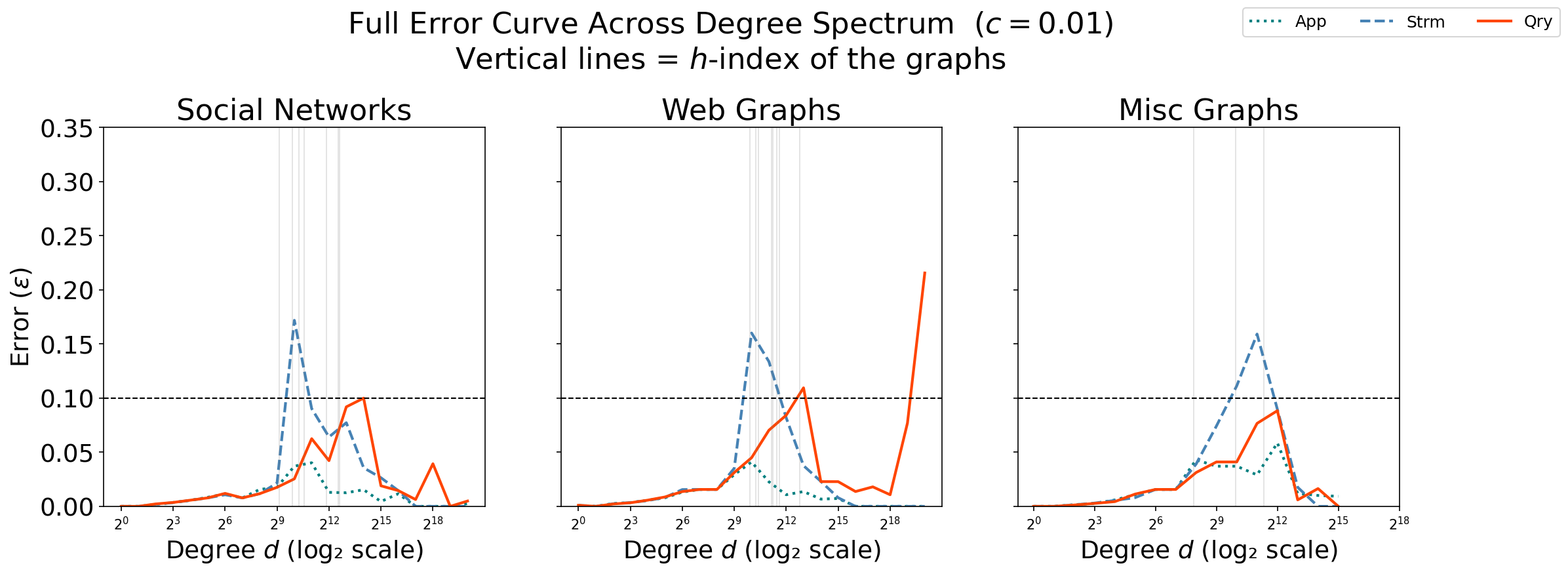}
        \caption{Error across entire degree spectrum for various types of graphs}
        \label{fig:P1_FullErrorCurve}
    \end{minipage}
    \begin{minipage}[b]{0.99\textwidth}
        \centering
        \includegraphics[width=\linewidth]{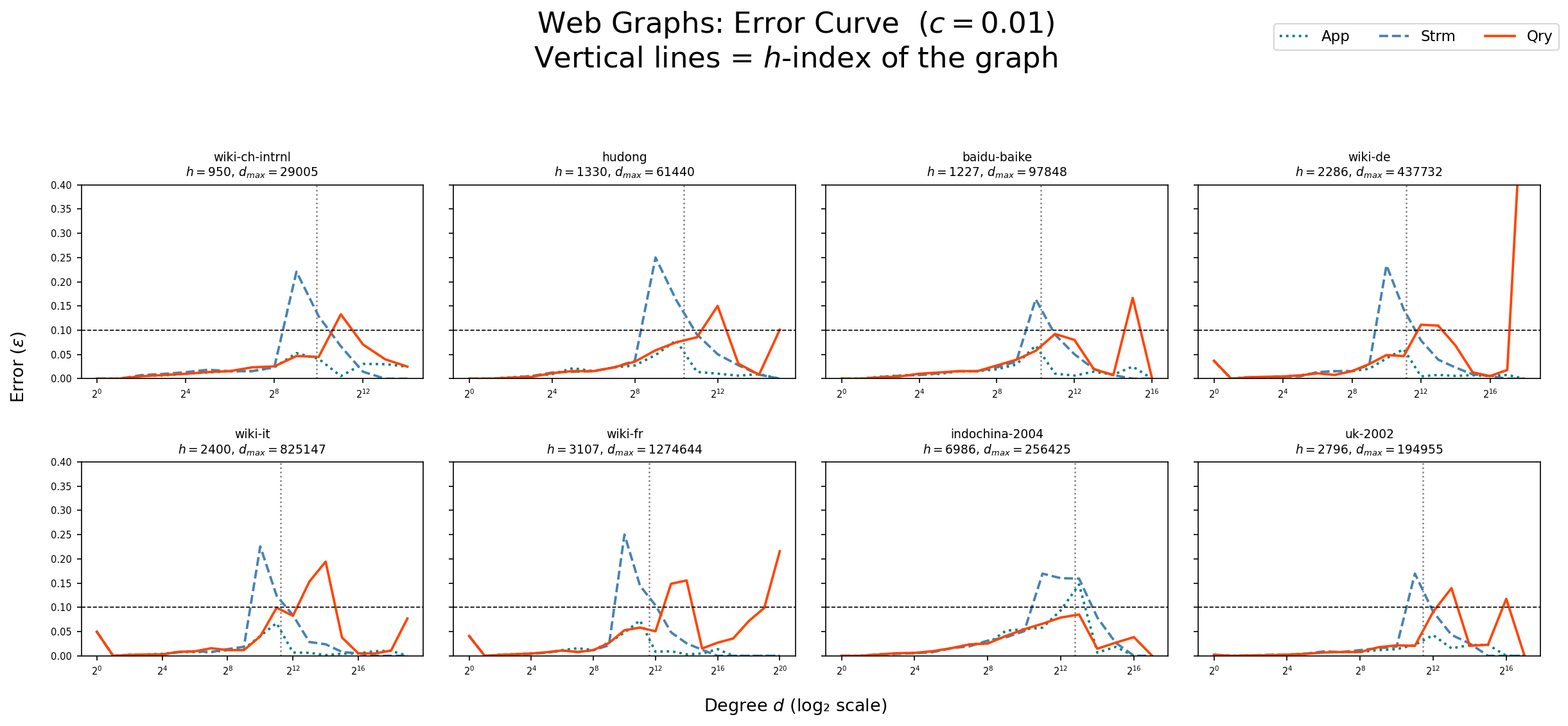}
        \caption{Error across entire degree spectrum for Web graphs}
        \label{fig:P8_WebSmallMultiples}
    \end{minipage}
    \begin{minipage}[b]{0.99\textwidth}
        \centering
        \includegraphics[width=\linewidth]{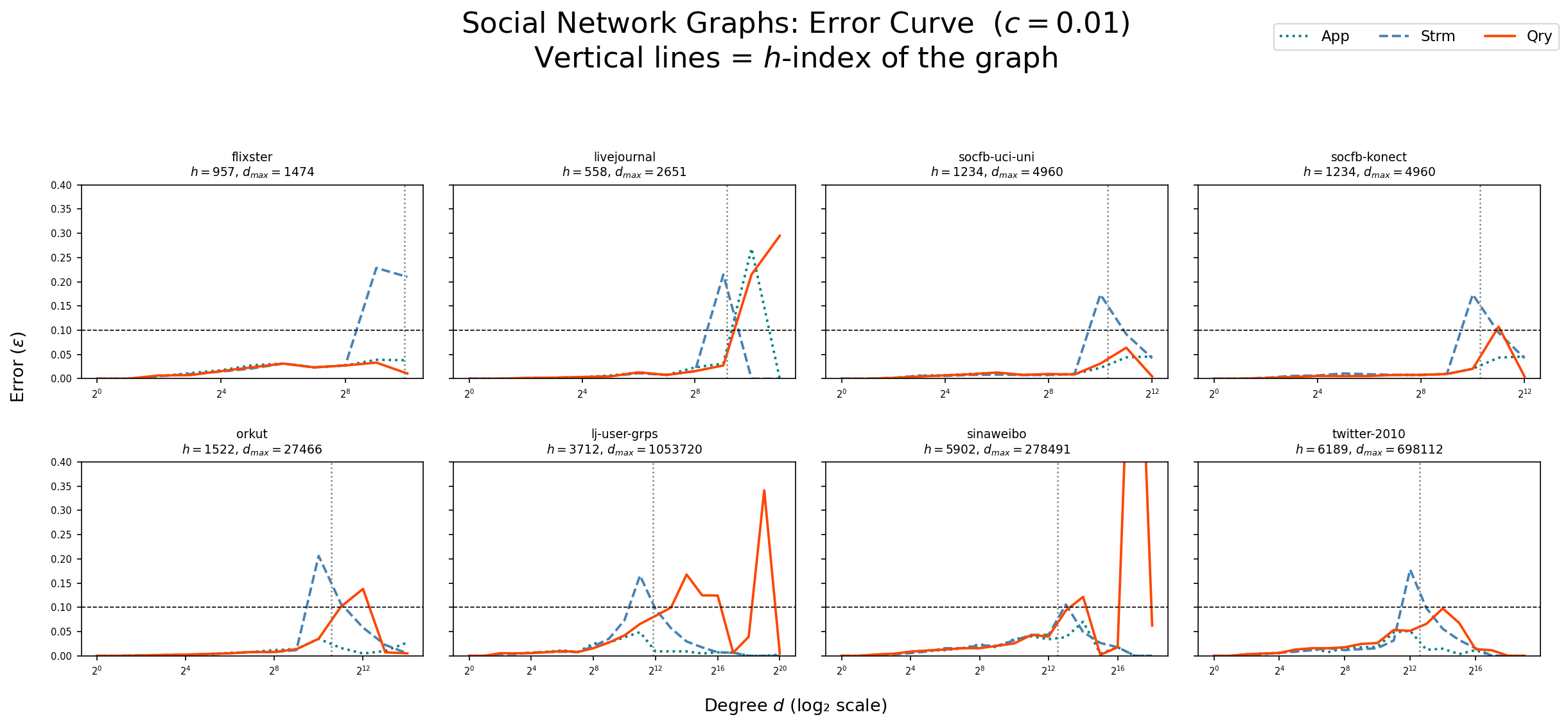}
        \caption{Error across entire degree spectrum for Social Network graphs}
        \label{fig:P9_SocSmallMultiples}
    \end{minipage}
    \caption*{Figures corresponding to \Cref{AppSubSec: Experiments Q3}}
\end{figure*}


\subsection{Regarding Q4}\label{AppSubSec: Experiments Q4}

\noindent\textbf{Q4:} What does the complexity of $\nicefrac{\edgecount}{\hindex}$ imply for real-world graphs?

To understand the implications of the term $\nicefrac{\edgecount}{\hindex}$ for various real world graphs, we study how $\edgecount$ and $\hindex$ vary with $\vertexcount$ for such graphs. In~\Cref{fig:HvsN}, we see that the h-index $\hindex$ broadly follows $\hindex = \sqrt{\vertexcount}$ for such graphs. It is also widely known~\cite{BarbasiAlbert/Science/1999/EmergenceofScalinginRandomNetworks,BroderKMRRSTW/CompNet/2000/GraphStructureinTheWeb,Leskovec/KDD/2005/GraphsOverTime,MeuselVLB/WWW/2014/GraphSturctWebRevisited,Eden/WWW/2017/CCDHinQueryModel,friedrich:LIPIcs.STACS.2019.5} that real world graphs are sparse. \Cref{fig:NvsM} substantiates that idea as the number of edges in the graphs studied are close to the reference line of $\edgecount = 15\vertexcount$. Combining these two observations, i.e. $\edgecount = \bigtheta{\vertexcount}$ and $\hindex = \bigtheta{\sqrt{\vertexcount}}$, the complexity $\nicefrac{\edgecount}{\hindex}$ seems to be $\bigtheta{\sqrt{\vertexcount}}$ for real world graphs.

\begin{figure*}[ht!]
    \centering
    
    \begin{minipage}{0.48\textwidth}
        \centering
        \includegraphics[width=\linewidth]{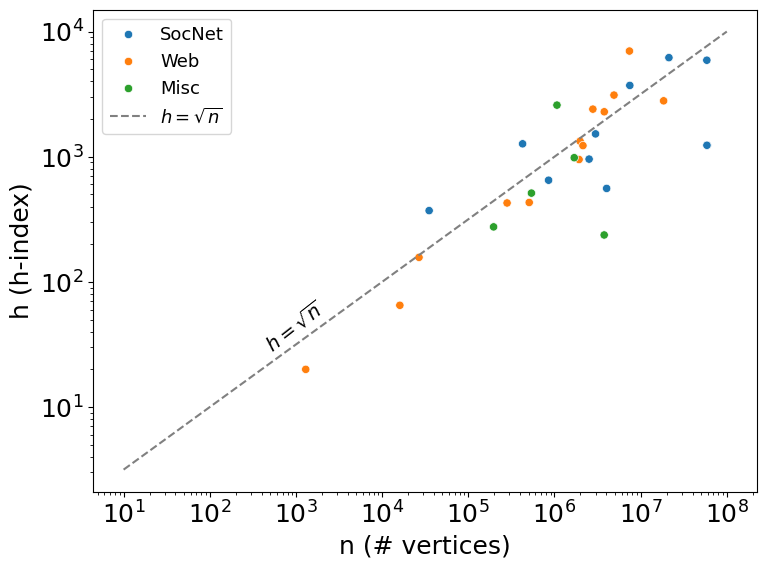}
        \caption{$\vertexcount$ vs $\hindex$ }
        \label{fig:HvsN}
        \end{minipage}\hfill
        \begin{minipage}{0.48\textwidth}
        \centering
        \includegraphics[width=\linewidth]{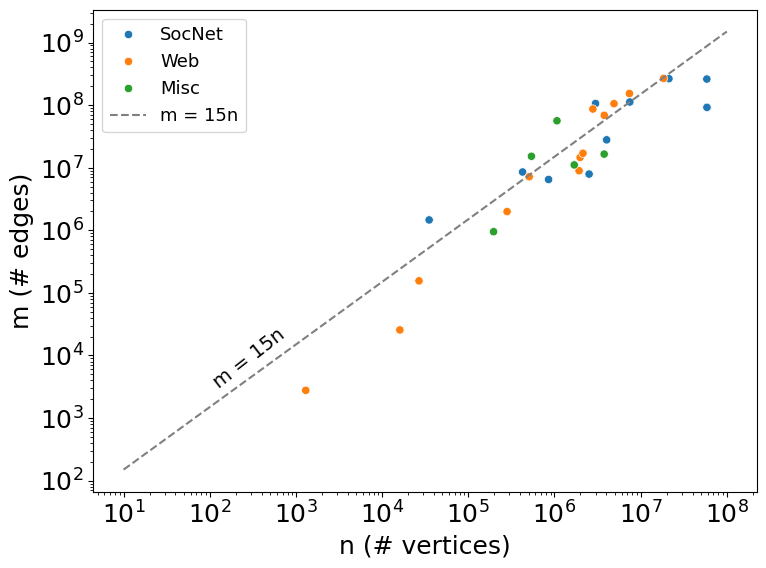}
        \caption{$\vertexcount$ vs $\edgecount$}
        \label{fig:NvsM}
    \end{minipage}
    \caption*{Figures corresponding to \Cref{AppSubSec: Experiments Q4}}
\end{figure*}

\end{document}